\documentclass[lettersize,journal]{IEEEtran}
\usepackage{amsmath,amsfonts}
\usepackage{algorithmicx}
\usepackage{algpseudocode}
\usepackage{algorithm}
\usepackage{array}
\usepackage[caption=false,font=normalsize,labelfont=sf,textfont=sf]{subfig}
\usepackage{textcomp}

\usepackage[autostyle=false, style=english]{csquotes}
\MakeOuterQuote{"}
\usepackage{stfloats}
\usepackage{amsmath,amssymb}
\usepackage{mathrsfs}
\usepackage{url}
\usepackage{verbatim}
\usepackage{graphicx}
\usepackage{cite}
\usepackage[mathlines]{lineno}
\usepackage{amsthm}
\usepackage{bbm}
\usepackage{bm}

\usepackage[dvipsnames]{xcolor}
\newtheorem{theorem}{Theorem}
\newtheorem{pro}{Proposition}
\newtheorem{lem}{Lemma}
\newtheorem{con}{Condition}

\newtheorem{rem}{Remark}
 
\DeclareMathOperator{\Tr}{Tr}

\DeclareMathOperator{\rank}{rank}

\DeclareMathOperator{\dist}{\mathrm{dist}}
\DeclareMathOperator{\sign}{sign}
\DeclareMathOperator{\ii}{\bm{\mathrm{i}}} 
\DeclareMathOperator{\jj}{\bm{\mathrm{j}}} 
\DeclareMathOperator{\kk}{\bm{\mathrm{k}}} 
 
\begin{document}

\title{Phase Retrieval of Quaternion 
  Signal via Wirtinger Flow}
\author{Junren Chen\thanks{The authors are with Department of 
		Mathematics, The University of Hong Kong, e-mail: chenjr58@connect.hku.hk; mng@maths.hku.hk}, Michael K. Ng,~\IEEEmembership{Senior Member,~IEEE.}
 }

\markboth{Accepted at "IEEE Transactions on  Signal Processing"}%
{Shell \MakeLowercase{\textit{et al.}}: A Sample Article Using IEEEtran.cls for IEEE Journals}


\maketitle

\begin{abstract}
The main aim of this paper is to study   quaternion phase retrieval (QPR), i.e., the recovery of quaternion signal from the magnitude of quaternion linear measurements. 
We show that all $d$-dimensional quaternion signals can be reconstructed
 up to a global right quaternion phase factor
from $O(d)$ phaseless measurements.
We also develop the scalable algorithm quaternion Wirtinger flow (QWF) for solving QPR, and establish its linear convergence guarantee. Compared with the analysis of complex Wirtinger flow, 
a series of different treatments are employed to overcome the difficulties of the non-commutativity of quaternion multiplication. 
 Moreover, we develop a variant of QWF that can effectively utilize a pure quaternion priori (e.g., for color images) by incorporating a quaternion phase factor estimate into QWF iterations. The estimate can be computed efficiently as it amounts to
finding a singular vector of a $4\times 4$ real matrix.
Motivated by  the variants of Wirtinger flow in prior work, we further propose quaternion truncated Wirtinger flow (QTWF), quaternion truncated amplitude flow (QTAF) and their pure quaternion versions.
Experimental results on synthetic data and color images are presented to validate our theoretical results. In particular, for pure quaternion signal recovery, our quaternion method often succeeds with notably fewer measurements  compared to real methods based on monochromatic model or concatenation model.
\end{abstract}


\section{Introduction}
\IEEEPARstart{A}{s} an expansion of the complex field $\mathbb{C}$, an element in the non-commutative field $\mathbb{Q} = \{q_0 + q_1\ii + q_2\jj + q_3\kk: q_0,q_1,q_2,q_3\in \mathbb{R}\}$ is called a quaternion number, which  contains one real part ($q_0$) and three imaginary parts ($q_1,q_2,q_3$). Although signals or images are traditionally  processed in $\mathbb{R}$ or $\mathbb{C}$, the quaternion algebra   has been noted to be a suitable platform for certain signal processing tasks. Consequently, many signal processing tools have been developed for quaternion setting over the past decades, including   Fourier transform \cite{sangwine1996fourier,ell2006hypercomplex}, wavelet transform \cite{bayro2006theory,fletcher2017development}, principal component analysis \cite{le2003quaternion}, moment analysis \cite{chen2015color,chen2012quaternion},   compressed sensing \cite{badenska2017compressed}, matrix completion \cite{jia2019robust,chen2022color}, deep neural network \cite{zhu2018quaternion,gaudet2018deep},   adaptive filtering \cite{took2008quaternion,ujang2011quaternion,took2010quaternion}, quaternion derivative \cite{mandic2010quaternion,xu2015enabling,xu2015optimization,xu2015theory}, and many others.

Color image processing is an important application  of quaternion. \textcolor{black}{To process color images in $\mathbb{R}$, one may use the monochromatic model that deals with each channel separately, or the concatenation model that concatenates three channels as a real matrix of triple size, whereas these two methods often fail to utilize the high correlations among channels. To fully utilize these correlations and process the color image as a whole, it was proposed to use tensor, quaternion or their integration for color image processing, see \cite{xu2023deep,jia2022non,xiao20192d,chen2019low} for instance.}  The quaternion-based approach encodes three  channels (i.e., Red, Green and Blue in the RGB color space) into three imaginary parts of pure quaternion, with 
the real part  set to  zero, hence the color image is modeled as pure quaternion signal. The advantage of quaternion-based approach is that the correlations among the three channels can be well preserved, and 
color images can be processed in a holistic manner. 
This approach was proposed in \cite{sangwine1996fourier,ell2006hypercomplex,pei1997novel}, and now has been extensively developed in various color imaging problems or methods, including denoising \cite{chen2019low,gai2015denoising,huang2021quaternion}, inpainting \cite{jia2019robust,chen2022color,miao2020quaternion}, segmentation  \cite{shi2007quaternion,subakan2011quaternion}, convolution neural network \cite{parcollet2019quaternion}, watermarking \cite{bas2003color,wang2013robust}, \textcolor{black}{sparse representation of color image \cite{xu2015vector,yu2013quaternion}}. By taking advantage of holistic processing of color images, quaternion-based approach usually outperforms processing methods in $\mathbb{R}$ (e.g., the aforementioned monochromatic model and concatenation model), see for instance \cite{chen2019low,gai2015denoising,huang2021quaternion,miao2020quaternion,
jia2019robust}.

Departing momentarily from the quaternion methods in signal processing, phase retrieval  concerning   signal reconstruction   from phaseless   measurements has attracted considerable research interest. It is motivated by a frequently encountered setting where it would be expensive, difficult, or even impossible to capture the measurement phase, to name a few,   X-ray crystallography \cite{millane1990phase}, quantum mechanics \cite{reichenbach1998philosophic}, speech recognition \cite{rabiner1993fundamentals}. We further emphasize the crucial role played by phase retrieval in many imaging problems like  diffraction imaging \cite{bunk2007diffractive}, astronomical imaging  \cite{fienup1987phase}, optics and microscopy \cite{walther1963question,miao2008extending}. For detailed discussion, we refer to 
the survey paper   \cite{shechtman2015phase}.
 
Mathematically, the goal of phase retrieval is to recover a 
signal $\bm{x}\in \mathbb{R}^d/\mathbb{C}^d$ based on the given knowledge of measurement matrix $\bm{A}\in \mathbb{R}^{n\times d}/\mathbb{C}^{n\times d}$ and the corresponding phaseless measurements $|\bm{Ax}|^2$ ($|\cdot|^2$ here applies element-wisely). From the theoretical side,   it is possible to reconstruct all signals   up to   a global phase factor (i.e., a sign $\pm 1$ in $\mathbb{R}$ or a unit complex scalar $\exp(\ii\theta)$ in $\mathbb{C}$) under optimal sample complexity $O(d)$, see \cite{balan2006signal,bandeira2014saving,conca2015algebraic,wang2019generalized} for instance. While most early algorithms for phase retrieval lack theoretical support \cite{gerchberg1972practical,fienup1978reconstruction,fienup1982phase}, a series of guaranteed algorithms have been developed in the past decade, which can be divided into a convex optimization approach \cite{candes2013phaselift} and a 
 non-convex optimization approach   \cite{candes2015phase,netrapalli2013phase,chen2017solving,wang2017solving}. 
Among them, the seminal work Wirtinger flow (WF) \cite{candes2015phase} provides a framework for algorithmic design under a non-convex optimization setting, i.e., via a careful initialization that can well approximate the solution, followed by gradient descent refinement. In many cases, this is essentially more practical and scalable than a convex lifting approach \cite{candes2013phaselift}.
 
Although quaternion is widely used to represent and process color images, and phase retrieval is a crucial technique in imaging science problems, 
quaternion phase retrieval (QPR) has not yet been studied before. More precisely, 
given a quaternion measurement matrix $\bm{A}\in \mathbb{Q}^{n\times d}$, QPR is concerned with the recovery of $\bm{x}\in \mathbb{Q}^d$ from the phaseless measurements $|\bm{Ax}|^2$. To our best knowledge, the only related result is presented in \cite{chen2022phase}, but it is restricted to  $\bm{A}\in \mathbb{R}^{n\times d}$. Under such real measurement matrix, the model fails to utilize 
the quaternion multiplication but simply identifies $\mathbb{Q}$ as  $\mathbb{R}^4$, thereby reducing to a special case of phase retrieval of real vector-valued signal. More prominently, compared with $\bm{A}\in \mathbb{Q}^{n\times d}$ studied in this work, $\bm{A}\in \mathbb{R}^{n\times d}$ leads to essentially more trivial ambiguities that can probably limit applications of QPR in signal processing (Remark \ref{realA}).

 The main aim of this paper is to close the research gap between  phase retrieval and quaternion signal processing. We initiate the study of QPR,  by first identifying the unavoidable trivial ambiguity, and then proposing and studying a practical algorithm of quaternion Wirtinger flow (QWF) with a linear convergence guarantee.     
This work is built upon many previous developments of quaternion, for instance, the  HR calculus of quaternion derivative \cite{mandic2010quaternion,xu2015theory}, and results for quaternion matrices \cite{zhang1997quaternions}. 
As color image processing is to deal with pure quaternion,
we also develop an algorithm that can incorporate the priori 
of a pure quaternion signal (i.e., $\Re(\bm{x})=0$) into QWF. 
Our main contributions are summarized as follows:
 \begin{itemize}
     \item {(Trivial Ambiguity)\textbf{.}} We show that all signals in $\mathbb{Q}^d$ can be reconstructed up to a global right quaternion phase factor from magnitude of $O(d)$ quaternion linear measurements (Theorem \ref{theorem1}). Moreover, in QPR of pure quaternion signals satisfying an extremely minor condition (three imaginary parts are real linearly independent), one can expect a reconstruction up to a sign of $\pm 1$ (Lemma \ref{lem66}).   
     
     \item {(Quaternion Wirtinger Flow)\textbf{.}} For solving QPR, we propose the QWF algorithm (Algorithm \ref{alg2}) consisting of spectral initialization and QWF refinement. Our main result (Theorem \ref{theorem2}) guarantees that using the magnitude of $O(d\log n)$  quaternion Gaussian measurements and under an error metric $\dist(\bm{x},\bm{y}) = \min_{\tt{q}\in\mathbb{T}_\mathbb{Q}}\|\bm{x}-\bm{y}\tt{q}\|$, the QWF sequence linearly converges to the underlying signal with high probability. Moreover, we propose a variant of QWF called pure quaternion Wirtinger flow (PQWF, Algorithm \ref{alg3}) for pure quaternion signal recovery. To utilize the pure quaternion priori, PQWF embeds an efficient quaternion phase factor estimate into the iteration and
     enjoys similar theoretical guarantee (Theorem \ref{theorem3}). The earlier phase transition is presented to confirm the efficacy of PQWF (Figure \ref{fig2}).

     \item (Refinements and Experiments)\textbf{.} Motivated by existing Wirtinger flow refinements for real/complex phase retrieval, we propose their counterparts in QPR (Algorithms \ref{alg4}-\ref{alg5}) and numerically show their improvements over QWF (Figure \ref{fig3}). We further specialize them to pure quaternion signal (Algorithms \ref{alg6}-\ref{alg7}) and then use them in color image recovery. Compared to   real phase retrieval based on monochromatic model or concatenation model, the proposed quaternion method succeeds with notably fewer phaseless measurements (Figures \ref{revisionadvan}-\ref{fig7}). 
 \end{itemize}

As part of our technical contributions, 
many essentially different   treatments 
take place in the proof
of Theorem \ref{theorem2} to overcome the challenges arising in quaternion setting. An evident example is the concentration of $\frac{1}{n}\sum_{k=1}^n \big[\Re(\bm{x^*\alpha}_k\bm{\alpha^*}_k\bm{h})\big]^2$,
where $\bm{\alpha^*}_k$ is the $k$-th row of $\bm{A}$.
In the proof, we avoid the  Hessian matrix employed in \cite{candes2015phase}, and instead calculate $\frac{1}{n}\sum_{k=1}^n \big[\Re(\bm{x^*\alpha}_k\bm{\alpha^*}_k\bm{h})\big]^2$ 
by using a real matrix representation $\mathcal{T}(\cdot)$ of a quaternion (Remark \ref{rem4}).
The formal definition of $\mathcal{T}(\cdot)$ can be found in
Section 2. 
In our analysis, such real matrix representations 
are recurring, e.g., 
  (\ref{3.24}), (\ref{add343}). 
Actually, with a great deal of quaternion-based ingredients involved in the theoretical analysis, we believe this work can technically provide an example for quaternion study and hence of some pedagogical value. 


This paper is organized as follows. In Section \ref{sec2} we state the notation and provide the preliminaries. Several useful techniques for quaternion study are included. 
In Section \ref{sec3} we first identify the trivial ambiguity in QPR,   then propose  QWF, and present the proof for its linear convergence. In Section \ref{sec4}, we propose PQWF for QPR of pure quaternion signal. 
Experimental results on both synthetic data and color image data are presented in Section \ref{sec6}. We give some remarks to conclude the paper in Section \ref{sec7}. To improve the readability, some auxiliary results for the main proof are provided in Appendix \ref{apenA}.


\section{Notations and Preliminaries}\label{sec2}

Some notations are needed for mathematical analysis. We denote probability and expectation respectively by $\mathbbm{P}(\cdot)$ and $\mathbbm{E}(\cdot)$, note that $\mathbbm{E}(\cdot)$ separately operates on one real part and three imaginary parts of a quaternion random variable. Besides, $C$, $c$, $C_i$, $c_i$ represent absolute constants whose value may vary from line to line. Both $T_1\lesssim T_2$ and $T_1= O(T_2)$ mean $T_1 \leq CT_2$ for some absolute constant $C$. Conversely, $T_1\geq CT_2$ is denoted by $T_1\gtrsim T_2$ or $T_1 = \Omega(T_2)$. We use capital boldface letters, lowercase boldface letters to   denote matrix, vector, respectively. \textcolor{black}{We denote real or complex scalar by regular letter (e.g., $a,b$), while quaternion scalar by the typewriter style letter (e.g., $\mathtt{a,b}$).} We conventionally write $[m] = \{1,\cdots,m\}$.

\subsection{Basics of Quaternions and Quaternion Matrices}
Let $\mathbb{Q} = \{\mathtt{q} = q_0+q_1\ii + q_2\jj + q_3\kk:q_0,q_1,q_2,q_3\in\mathbb{R}\}$ be the set of quaternion numbers, then $\mathbb{Q}^d$ (resp. $\mathbb{Q}^{d_1\times d_2}$) denotes the set of $d$-dimensional quaternion vectors (resp. $n_1\times n_2$ quaternion matrices). The addition and subtraction of two quaternion numbers are defined component-wisely, e.g., $(a\ii+b\kk) - (c\jj + d\ii)=(a-d)\ii -c \jj + b \kk$. Under the rules $\ii^2 = \jj^2 = \kk^2 = -1$, $\ii\jj = -\jj\ii = \kk$, $\jj\kk = -\kk\jj = \ii$, $\kk\ii = -\ii\kk = \jj$ and with distributive law, associative law imposed, the multiplication between quaternion numbers is defined. 

Given  $\mathtt{q}=q_0+q_1\ii+q_2\jj+q_3\kk$, besides the real part $\Re (\mathtt{q}) = q_0$ and the vector part $\Im(\mathtt{q}) = q_1\ii + q_2\jj + q_3\kk$, we further use $\mathcal{P}^\vartheta(\vartheta = \ii,\jj,\kk)$ to extract three imaginary components, i.e., $\mathcal{P}^{\ii}(\mathtt{q})=q_1$, $\mathcal{P}^{\jj}(\mathtt{q})=q_2$, $\mathcal{P}^{\kk}(\mathtt{q})=q_3$. Note that $\overline{\mathtt{q}} = q_0 - q_1\ii-q_2\jj-q_3\kk$ is the conjugate of $\mathtt{q}$, $|\mathtt{q}| = ({\sum_{k=0}^3 q_k^2})^{1/2}$ is the absolute value. We allow these operations for quaternion number element-wisely apply to 
quaternion vectors and matrices. For nonzero $\mathtt{q}$,   $\mathtt{q}^{-1} = \overline{\mathtt{q}}/|\mathtt{q}|^2$ is its inverse. As pure quaternions with zero real part is of particular interest, we collect them in $\mathbb{Q}_p = \{\mathtt{q}\in \mathbb{Q}:\Re(\mathtt{q})=0\}$. The phase of a non-zero quaternion $q$ is defined as $\sign(\mathtt{q})= \frac{\mathtt{q}}{|\mathtt{q}|}$. As the phase belongs to $\mathbb{T}_\mathbb{Q}=\{\mathtt{q}\in \mathbb{Q}:|\mathtt{q}|=1\}$,   we sometimes call quaternion in $\mathbb{T}_\mathbb{Q}$ the quaternion phase factor, e.g., when we describe the trivial ambiguity.

Given the vector $\bm{x} = [\mathtt{x}_k]\in \mathbb{Q}^d$ or the matrix $\bm{A}=[\mathtt{a}_{ij}]\in \mathbb{Q}^{d_1\times d_2} $, we introduce the (vector) $\ell_2$ norm $\|\bm{x}\|= ({\sum_k |\mathtt{x}_k|^2})^{1/2}$, the matrix operator norm $\|\bm{A}\| = \sup_{\bm{w}\in \mathbb{Q}^{d_1}\setminus\{\bm{0}\}}\|\bm{Aw}\|/\|\bm{w}\|$, the matrix Frobenius norm $\|\bm{A}\|_F = ({\sum_{i,j}|\mathtt{a}_{ij}|^2})^{1/2}$. The rank of $\bm{A}$, denoted $\rank(\bm{A})$, is defined to be the maximum number of right linearly independent columns of $\bm{A}$. Note that quaternion vectors $\bm{\alpha}_1,\cdots,\bm{\alpha}_N$ are said to be right linearly independent if $\sum_{k=1}^N \bm{\alpha}_k \mathtt{q}_k =0~ (\mathtt{q}_k\in \mathbb{Q}) $ can imply $\mathtt{q}_k =0$ for all $k$. Let $\bm{I}_d$ be the identity matrix, the matrix $\bm{A}\in\mathbb{Q}^{d\times d}$ is invertible if there exists $\bm{B}$ such that $\bm{AB}=\bm{BA}=\bm{I}_d$. Parallel to  complex matrices, $\bm{A}$ is invertible if and only if it is full rank ($\rank(\bm{A})=d$) \cite{zhang1997quaternions}.  We say $\bm{A}\in \mathbb{Q}^{d\times d}$ is Hermitian if $\bm{A^*}=\bm{A}$, or is unitary if $\bm{AA^*} =\bm{A^*A} = \bm{I}_d$. We use $\mathcal{H}^\mathbb{Q}_{d,r}$ to denote the set of all $d\times d$ Hermitian matrices with rank not exceeding $r$. The standard inner product   for quaternion vector or matrix is given by $\big<\bm{A},\bm{B}\big> = \Tr (\bm{A^*B})$, where $\Tr(.)$ returns the sum of diagonal entries for a square matrix, or simply the scalar itself. Evidently, $\Re\big<\bm{A},\bm{B}\big>=\Re\big<\bm{B},\bm{A}\big>$.

Note that quaternion multiplication is non-commutative (e.g., $\ii\jj = -\jj\ii$), which is often a key technical challenge in the extension  from $\mathbb{R}$/$\mathbb{C}$ to quaternion setting. We note that, taking the real part is an effective technique to circumvent this issue, since we have $\Re(\mathtt{ab})=\Re(\mathtt{ba})$ for any $\mathtt{a,b}\in\mathbb{Q}$, or more generally, for $\bm{A}\in \mathbb{Q}^{d_1\times d_2},\bm{B}\in\mathbb{Q}^{d_2\times d_1}$ it holds that
\begin{equation}
    \begin{aligned}
    \Re\big(\Tr(\bm{AB})\big) =\Re\big(\Tr(\bm{BA})\big).
    \end{aligned}
\end{equation}

\subsection{Real representation and Quaternion SVD}
We further introduce two key techniques to study quaternion matrices. The first one is the  complex or  real representation of quaternion matrices based on the maps $\mathcal{T}_{\mathbb{C}}(\cdot)$, $\mathcal{T} (\cdot)$. Note that any $\bm{A}\in\mathbb{Q}^{d_1\times d_2}$ can be uniquely written as $\bm{A} = \bm{B} + \bm{C}\jj$ for some $\bm{B} ,\bm{C}\in\mathbb{C}^{d_1\times d_1}$, then $\mathcal{T}_{\mathbb{C}}(\cdot)$ maps $\bm{A}$ to its complex adjoint matrix belonging to $\mathbb{C}^{2d_1\times 2d_2}$  
\begin{equation}
    \nonumber
    \mathcal{T}_\mathbb{C} (\bm{A}):= \begin{bmatrix} \bm{B} & \bm{C} \\
    -\overline{\bm{C}} & \overline{\bm{B}}
\end{bmatrix}.
\end{equation}
For quaternion matrices $\bm{A},\bm{A}_1,\bm{A}_2$,  several useful relations are in order: $$\mathcal{T}_{\mathbb{C}}(\bm{A}_1\bm{A}_2) =  \mathcal{T}_{\mathbb{C}}(\bm{A}_1)  \mathcal{T}_{\mathbb{C}}(\bm{A}_2);$$
$$\mathcal{T}_{\mathbb{C}}(\bm{A}_1+\bm{A}_2) =  \mathcal{T}_{\mathbb{C}}(\bm{A}_1)+\mathcal{T}_{\mathbb{C}}(\bm{A}_2);$$
$$\mathcal{T}_{\mathbb{C}}(\bm{A^*}) = \big(\mathcal{T}_{\mathbb{C}}(\bm{A})\big)^*;$$
see Theorem 4.2 of \cite{zhang1997quaternions}. Furthermore,   $\bm{A}\in \mathbb{C}^{d_1\times d_2}$ can be written as $\bm{A}=\bm{B}+\bm{C}\ii$ for some $\bm{B},\bm{C}\in\mathbb{R}^{d_1\times d_2}$,  which  can then be  reduced   to  real matrix by $\mathcal{T}_{\mathbb{R}}$: 
\begin{equation}
    \nonumber
    \mathcal{T}_{\mathbb{R}}(\bm{A}) :=\begin{bmatrix} \bm{B} & \bm{C} \\
    -{\bm{C}} & {\bm{B}}
\end{bmatrix},
\end{equation}
 and for complex matrices $\bm{A}$, $\bm{A}_1$, $\bm{A}_2$ one also has $$\mathcal{T}_{\mathbb{R}}(\bm{A}_1\bm{A}_2) =  \mathcal{T}_{\mathbb{R}}(\bm{A}_1)  \mathcal{T}_{\mathbb{R}}(\bm{A}_2);$$ 
 $$\mathcal{T}_{\mathbb{R}}(\bm{A}_1+\bm{A}_2) =  \mathcal{T}_{\mathbb{R}}(\bm{A}_1)+\mathcal{T}_{\mathbb{R}}(\bm{A}_2);$$
 $$\mathcal{T}_{\mathbb{R}}(\bm{A}^*) = \big(\mathcal{T}_{\mathbb{R}}(\bm{A})\big)^\top.$$
Naturally, a composition of these two maps, i.e., $\mathcal{T}:= \mathcal{T}_{\mathbb{R}}\circ \mathcal{T}_{\mathbb{C}}$, can reduce $\bm{A}\in\mathbb{Q}^{d_1\times d_2}$ to $\mathcal{T}(\bm{A})\in \mathbb{R}^{4d_1\times 4d_2}$. More precisely,  \begin{equation}
\nonumber
    \mathcal{T}(\bm{A}) = \begin{bmatrix}
    \Re(\bm{A}) & \mathcal{P}^{\jj}(\bm{A})&\mathcal{P}^{\ii}(\bm{A})&\mathcal{P}^{\kk}(\bm{A}) \\
    -\mathcal{P}^{\jj}(\bm{A}) & \Re(\bm{A}) & \mathcal{P}^{\kk}(\bm{A}) & -\mathcal{P}^{\ii}(\bm{A})\\
   -\mathcal{P}^{\ii}(\bm{A})&-\mathcal{P}^{\kk}(\bm{A}) &\Re(\bm{A}) & \mathcal{P}^{\jj}(\bm{A})\\
  - \mathcal{P}^{\kk}(\bm{A}) &  \mathcal{P}^{\ii}(\bm{A}) & -\mathcal{P}^{\jj}(\bm{A}) & \Re(\bm{A}) 
    \end{bmatrix}.
\end{equation} The following relations hold due to the properties of $\mathcal{T}_\mathbb{C}$ and $\mathcal{T}_\mathbb{R}$: $$\mathcal{T}(\bm{A}_1\bm{A}_2) = \mathcal{T}(\bm{A}_1)\mathcal{T}(\bm{A}_2);$$
$$\mathcal{T}(\bm{A}_1+\bm{A}_2)=\mathcal{T}(\bm{A}_1)+\mathcal{T}(\bm{A}_2);$$
$$\mathcal{T}(\bm{A}^*) = \big(\mathcal{T}(\bm{A})\big)^\top.$$
We would also define $\mathcal{T}_i(\bm{A})$ ($i=1,2,3,4$) to be the $i$-th column of blocks in $\mathcal{T}(\bm{A})$. For instance, $$\mathcal{T}_1(\bm{A}):= [\Re(\bm{A})^\top,-\mathcal{P}^{\jj}(\bm{A})^\top,-\mathcal{P}^{\ii}(\bm{A})^\top,-\mathcal{P}^{\kk}(\bm{A})^\top]^\top\in \mathbb{R}^{4d_1\times d_2}.$$ Also note the relation     $\mathcal{T}_i(\bm{A}_1\bm{A}_2) = \mathcal{T}(\bm{A}_1)\mathcal{T}_i(\bm{A}_2)$.

In addition, quaternion singular value decomposition (QSVD)   is another powerful tool.  For $\bm{A}\in \mathbb{Q}^{d_1\times d_2}$, there exist unitary matrices $\bm{U}\in \mathbb{Q}^{d_1\times d_1}$, $\bm{V}\in \mathbb{Q}^{d_2\times d_2}$,  diagonal matrix $\bm{\Sigma}\in \mathbb{R}^{d_1\times d_2}$ with non-negative diagonal entries $\sigma_1,\cdots ,\sigma_{\min\{d_1,d_2\}}$, such that $\bm{A} = \bm{U\Sigma V^*}$. We refer readers to    \cite[Theorem 7.2]{zhang1997quaternions}  for its derivation. In QSVD, $\{\sigma_k:1\leq k\leq \min\{d_1,d_2\}\}$ are the singular values of $\bm{A}$. It is straightforward to show several facts coincident with $\mathbb{R}^{d_1\times d_2}$ or $\mathbb{C}^{d_1\times d_2}$, e.g., 
the maximum singular value equals $\|\bm{A}\|$, and $\|\bm{A}\|_F = (\sum_{k}\sigma_k^2)^{1/2}$. Moreover,     the  SVD for   $\mathcal{T}(\bm{A})$ is given by $\mathcal{T}(\bm{A}) = \mathcal{T}(\bm{U})\mathcal{T}(\bm{\Sigma})\mathcal{T}(\bm{V})^\top$, which leads to $\|\mathcal{T}(\bm{A})\| =\|\bm{A}\|$, $\|\mathcal{T}(\bm{A})\|_{F} = 2\|\bm{A}\|_F$. 
We will also work with the matrix nuclear norm defined to be the sum of singular values, i.e., $\|\bm{A}\|_{nu} = \sum_k \sigma_k$.
  
  \subsection{(Standard) Eigenvalue and Eigenvector}
    For simplicity, in this paper we restrict the eigenvalue to be the right one.\footnote{However, in a complete theory, left eigenvalue and right eigenvalue for quaternion matrices should be distinguished \cite{zhang1997quaternions}.} In particular, given $\bm{A}\in \mathbb{Q}^{d\times d}$, if $\bm{Ax} = \bm{x}\mathtt{\lambda}$ for some nonzero $\bm{x}\in \mathbb{Q}^d$, we refer $\mathtt{\lambda}$, $\bm{x}$ to as the  eigenvalue, eigenvector  of $\bm{A}$. Since $\bm{Ax} = \bm{x}\mathtt{\lambda}$ is equal to $\bm{A } (\bm{x}\mathtt{v}^*)= (\bm{x}\mathtt{v}^*)(\mathtt{v\lambda v}^*)$ for any   $\mathtt{v}\in \mathbb{T}_\mathbb{Q}$, $\bm{A}$ with eigenvalue $\lambda$ indeed possesses a set of eigenvalues $\{\mathtt{v\lambda v}^*:\mathtt{v}\in \mathbb{T}_\mathbb{Q}\}$, among which we can pick a unique "standard eigenvalue" in the form of $a+b\ii$ ($a\in\mathbb{R}$, $b\geq 0$), see   \cite[Lemma 2.1]{zhang1997quaternions}. Any $\bm{A}\in\mathbb{Q}^{d\times d}$ has exactly $d$ standard eigenvalues, and particularly, all standard eigenvalues of Hermitian $\bm{A}$ are real. Akin to the eigenvalue decomposition for complex Hermitian   matrices, quaternion Hermitian matrix  $\bm{A}$ can be decomposed as $\bm{A} = \bm{U\Sigma U^*}$ for some unitary $\bm{U}$ and diagonal matrix $\bm{\Sigma}$, with standard eigenvalues of $\bm{A}$ arranged in the diagonal of $\bm{\Sigma}$  \cite[Corollary 6.2]{zhang1997quaternions}.

  \subsection{Derivative of Real     Function with Quaternion Variable}
  We need to calculate the derivative of real function with quaternion variable. The framework that best meets such optimization need is the HR calculus, see \cite{mandic2010quaternion} and the more complete theory in \cite{xu2015enabling,xu2015optimization,xu2015theory}. More precisely, given a function $f(\bm{q})$ with quaternion variable $\bm{q} = \bm{q}_a + \bm{q}_b \ii + \bm{q}_c \jj + \bm{q}_d \kk\in \mathbb{Q}^{N}$, we define \begin{equation}
  \label{2.7}
      \frac{\partial f}{\partial \bm{q}} : = \frac{1}{4}\Big(\frac{\partial f}{\partial \bm{q}_a} -  \frac{\partial f}{\partial \bm{q}_b}\ii - \frac{\partial f}{\partial \bm{q}_c} \jj -  \frac{\partial f}{\partial \bm{q}_d}\kk\Big) \in \mathbb{Q}^N.
  \end{equation} 
 Generally, this is called left derivative  and is different from the right derivative $$\frac{\partial f}{\partial \bm{q}} = \frac{1}{4}\big(\frac{\partial f}{\partial \bm{q}_a} -\ii  \frac{\partial f}{\partial \bm{q}_b}  -\jj  \frac{\partial f}{\partial \bm{q}_c}  -  \kk\frac{\partial f}{\partial \bm{q}_d}\big).$$  Considering only the derivative of real function $f$ ($f\in\mathbb{R}$) will be involved in this work, we simply adopt the left one (\ref{2.7}). For $f\in \mathbb{R}$ with quaternion variable $\bm{q}$, the gradient \begin{equation}
     \nabla f(\bm{q})=\Big(\frac{\partial f}{\partial \bm{q}}\Big)^*=\frac{1}{4} \Big(\frac{\partial f}{\partial \bm{q}_a} +  \frac{\partial f}{\partial \bm{q}_b}\ii + \frac{\partial f}{\partial \bm{q}_c} \jj + \frac{\partial f}{\partial \bm{q}_d}\kk\Big)
     \label{gradientt}
 \end{equation} represents the direction in which $f$ changes in a maximum rate \cite{mandic2010quaternion,xu2015theory}. Some calculation rules   in \cite{xu2015theory}  would also be used later. 

\section{Quaternion Wirtinger Flow}\label{sec3}

\subsection{The Trivial Ambiguity}

Given a quaternion signal $\bm{x}\in \mathbb{Q}^d$ and a measurement matrix $\bm{A}\in \mathbb{Q}^{n\times d}$, in QPR we aim to reconstruct $\bm{x}$ from $\{|\bm{\alpha^*}_k\bm{x}|^2:k\in [n]\}$, where $\bm{\alpha^*}_k$ is the $k$-th row of $\bm{A}$.
Observe that for any unit quaternion $\mathtt{q}$ (i.e., $\mathtt{q}\in\mathbb{T}_\mathbb{Q}$), $|\bm{A}(\bm{x}\mathtt{q})|^2 = |\bm{Ax}|^2$,  so one can never distinguish two signals   only differentiated by a global right quaternion phase factor.\footnote{For nonzero quaternion $q$ we call $\frac{q}{|q|}$   its phase. Hence, we will refer to $q\in \mathbb{T}_{\mathbb{Q}}$ as a quaternion phase factor.} However, a global left quaternion phase factor $\mathtt{q}$ is not necessarily an ambiguity since $\mathtt{q}$ and $\bm{A}$ may not be commutative, hence $|\bm{A}(\mathtt{q}\bm{x})|^2 \neq |\bm{Ax}|^2$ is possible, see more discussions below Theorem \ref{theorem1}.\footnote{However, when a real measurement matrix is used (as in \cite{chen2022phase}), this becomes a trivial ambiguity due to $|\bm{A}(\mathtt{q}\bm{x})|^2 =|\mathtt{q}\bm{Ax}|^2=|\bm{Ax}|^2$ when $\bm{A}\in\mathbb{R}^{n\times d}$. This is the essential difference between this work and the QPR result in \cite{chen2022phase}.} This is in stark contrast to the real or complex phase retrieval.

Throughout this paper we consider the Gaussian measurement ensemble where the entries of $\bm{A}$ are i.i.d. drawn from $$\mathcal{N}_\mathbb{Q}:=\frac{1}{2}\mathcal{N}(0,1)+\frac{1}{2}\mathcal{N}(0,1)\ii+\frac{1}{2}\mathcal{N}(0,1)\jj+\frac{1}{2}\mathcal{N}(0,1)\kk,$$ denoted by $\bm{A}\sim \mathcal{N}_\mathbb{Q}^{n\times d}$. Evidently, $\mathbbm{E}(\bm{\alpha_k\alpha_k^*})=\bm{I}_d$. 

Generally speaking, the goal in any signal reconstruction task is to recover the signal up to trivial ambiguity. Thus, a question of fundamental importance is whether there exist other unavoidable ambiguities in QPR (besides the aforementioned right phase factor).  In the next theorem,   we show that the global right quaternion phase factor is the only trivial ambiguity by proving a stronger uniform recovery guarantee: all signals in $\mathbb{Q}^d$ can be reconstructed up to right quaternion phase factor from  $O(d)$ phaseless measurements.

\begin{theorem}
\label{theorem1}
Assume $\bm{A}=[\bm{\alpha}_1,\cdots,\bm{\alpha}_n]^*\sim \mathcal{N}_\mathbb{Q}^{n\times d}$. When $n\geq Cd$ for some absolute constant $C$, with probability at least $1-\exp(-C_1n)$, all signals $\bm{x}$  in $\mathbb{Q}^d$ can be reconstructed from $\{|\bm{\alpha^*}_k\bm{x}|^2:k\in [n]\}$ up to a global right quaternion phase factor.  
\end{theorem}

\begin{proof}
    The proof can be found in supplementary material.
\end{proof}

\textcolor{black}{
  In the proof of Theorem \ref{theorem1},  we use Mendelson's small ball method to show that, with high probability,  $|\bm{\alpha}_k^*\bm{x}|^2=|\bm{\alpha}_k^*\bm{y}|^2~(\forall k\in [n])$ implies $\bm{xx}^*=\bm{yy}^*$  (while the converse statement is evidently true). As shown in Lemma \ref{trialtri},   $\bm{xx}^*=\bm{yy}^*$ is equivalent to $\bm{x}=\bm{y}\mathtt{q}$ for some $\mathtt{q}\in \mathbb{T}_\mathbb{Q}$.}

  \textcolor{black}{
  Note that a global left   phase factor is not trivial ambiguity, as for  $\mathtt{q}\in \mathbb{T}_\mathbb{Q}$, $(\mathtt{q}\bm{x})(\mathtt{q}\bm{x})^*=\mathtt{q}\bm{x x}^*\bar{\mathtt{q}}$ does not equal to $\bm{xx}^*$ in general. Indeed, even it happens that $(\mathtt{q}\bm{x})(\mathtt{q}\bm{x})^*=\bm{xx}^*$, by Lemma \ref{trialtri}, such ambiguity can be expressed via a right quaternion phase factor (i.e., $\mathtt{q}\bm{x}=\bm{x}\mathtt{q}_1$ for some $\mathtt{q}_1\in \mathbb{T}_\mathbb{Q}$). 
 \begin{lem}
     \label{trialtri}
     Let $\bm{x},\bm{y}\in \mathbb{Q}^d$, then $\bm{xx}^*=\bm{yy}^*$ is equivalent to $\bm{x}=\bm{y}\mathtt{q}$ for some $\mathtt{q}\in\mathbb{T}_\mathbb{Q}$.
 \end{lem} 
\begin{proof}
If $\bm{x}=\bm{y}\mathtt{q}$ for some unit quaternion $\mathtt{q}$, then $\bm{xx^*}=\bm{y}\mathtt{q}\bar{\mathtt{q}}\bm{y}^*=\bm{yy}^*$. Thus, it remains to prove $\bm{x}=\bm{y}\mathtt{q}$ ($\exists  \mathtt{q}\in\mathbb{T}_\mathbb{Q}$) from $\bm{xx^*}=\bm{yy^*}$.
We let $\bm{x}=[\mathtt{x}_i],\bm{y}=[\mathtt{y}_i]$, then we have $\mathtt{x}_i\overline{\mathtt{x}_j}=\mathtt{y}_i\overline{\mathtt{y}_j}$ for any $i,j\in [d]$. Let $i=j$, we obtain $|\mathtt{x}_i|=|\mathtt{y}_i|$, so we can assume $\mathtt{x}_i=\mathtt{y}_i\mathtt{q}_i$ for some $\mathtt{q}_i\in \mathbb{T}_\mathbb{Q}$. For $i\neq j$, we thus have $\mathtt{y}_i\mathtt{q}_i\overline{\mathtt{q}_j}\overline{\mathtt{y}_j}=\mathtt{y}_i\overline{\mathtt{y}_j}$. Assuming $\mathtt{y}_i,\mathtt{y}_j$ are both non-zero, this implies $\mathtt{q}_i=\mathtt{q}_j=\mathtt{q}$ for some common $\mathtt{q}$. When $\mathtt{y}_i=0$, then $\mathtt{x}_i=0$, we also have $\mathtt{x}_i=\mathtt{y}_i\mathtt{q}$. Therefore, there exists a common $\mathtt{q}\in \mathbb{T}_{\mathbb{Q}}$ such that $\bm{x}=\bm{y}\mathtt{q}$.  
\end{proof}
  }

Note that using Mendelson's small ball method with a little bit more work  one can prove uniform stable recovery guarantee \cite{chen2022error}.

We further give a remark comparing our model and QPR with real measurement matrix studied in \cite{chen2022phase}.

\begin{rem}\label{realA}\textcolor{black}{
    Some results on recovering $\bm{x}\in \mathbb{Q}^d$ from $|\bm{Ax}|^2$ with  $\bm{A}\in \mathbb{R}^{n\times d}$ were presented in \cite{chen2022phase}. Using a real measurement matrix, such model does not really utilize the special quaternion multiplication because for $\bm{\alpha}\in\mathbb{R}^d,\bm{x}\in \mathbb{Q}^d$ we have $|\bm{\alpha}^\top\bm{x}|=\|\bm{\alpha}^\top\bm{x}'\|$, where $\bm{x}'=[\Re\bm{x},\mathcal{P}^{\ii}\bm{x},\mathcal{P}^{\jj}\bm{x},\mathcal{P}^{\kk}\bm{x}]\in \mathbb{R}^{d\times 4}$ can be viewed as a $d$-dimensional $\mathbb{R}^4$-valued vector. Thus, it simply identifies $\mathbb{Q}$ with $\mathbb{R}^4$. Compared to our QPR model, the downside of such model is that it suffers from much more trivial ambiguities, e.g., the left quaternion phase factor as $|\bm{A}q\bm{x}|^2=|\bm{Ax}|^2$ holds for $\mathtt{q}\in \mathbb{T}_\mathbb{Q}$, the conjugate as $|\bm{A}\bar{\bm{x}}|^2=|\bm{Ax}|^2$, and moreover a $4\times 4$ orthogonal matrix operating on the real part and three imaginary parts because  $\|\bm{\alpha}^\top\bm{x}'\|=\|\bm{\alpha}^\top\bm{x}'\bm{O}\|$ holds for any $4\times 4$ orthogonal matrix $\bm{O}$ (interested readers can verify that this ambiguity is already more severe than the right quaternion phase factor in our model).  Indeed, we will show   that using our QPR model, most pure quaternion signals can be reconstructed up to a sign (Lemma \ref{lem66}), but this is not possible under real  measurement matrix because of the trivial ambiguity of a $3\times 3$ orthogonal matrix operating on three imaginary parts (as an example, let $\bm{a},\bm{b}\in \mathbb{R}^d$ we cannot distinguish $\bm{a}\ii+\bm{b}\jj$ and $\bm{b}\ii+\bm{a}\jj$ using the phaseless measurements produced by a real measurement matrix).}
\end{rem}

\subsection{The Quaternion Wirtinger Flow Algorithm}\label{sub3.1}

Recall that WF     for solving real/complex phase retrieval problem   contains spectral initialization and WF update as two steps \cite{candes2015phase}, and we will first present the QWF algorithm in analogy. We will exclusively use $\bm{x}$   to denote the underlying quaternion signal and assume $\|\bm{x}\|=1$. For succinctness we focus on noiseless case where the $k$-th measurement is 
      $y_k = |\bm{\alpha^*}_k \bm{x}|^2$. 
 
 

The QWF algorithm is based on minimizing the $\ell_2$ loss \begin{equation}
    \label{l2goal}
    \min_{\bm{z}\in \mathbb{Q}^d} f(\bm{z}):= \frac{1}{n}\sum_{k=1}^n \big(|\bm{\alpha^*}_k\bm{z}|^2-y_k \big)^2. 
\end{equation} 
We first describe the careful initialization by spectral method in Algorithm \ref{alg1}. Intuitively, $\bm{\nu}_{in}$ can well approximate the direction of $\bm{x}$ due to $\mathbbm{E}\bm{S}_{in} =\bm{I}_d + \frac{1}{2}\bm{xx^*}$ (Lemma \ref{lem3}(d)). Because $\mathbbm{E}y_k = \|\bm{x}\|^2$, it is natural     to estimate the signal norm $\|\bm{x}\|$ as $\lambda_0=\big(\frac{1}{n}\sum_{k=1}^n y_k\big)^{1/2}$.

\begin{algorithm} 
\caption{\textcolor{black}{Spectral Initialization}}\label{alg1}
\begin{algorithmic}[1]
 \Statex \textbf{Input:} data $(\bm{\alpha}_k,y_k)_{k=1}^n$

 \item  Construct the Hermitian data matrix \begin{equation}\label{datam}
    \bm{S}_{in} = \frac{1}{n}\sum_{k=1}^n y_k \bm{\alpha}_k\bm{\alpha^*}_k
\end{equation}
and find its normalized eigenvector regarding the largest standard eigenvalue $\bm{\nu}_{in}$. 

\item  We compute $\lambda_0:=\big(\frac{1}{n}\sum_{k=1}^n y_k\big)^{1/2}$ and obtain  the spectral initialization
$
    \bm{z}_0=\lambda_0\cdot \bm{\nu}_{in}.$
\Statex\textbf{Output:} $\bm{z}_0$
\end{algorithmic}
\end{algorithm}

Then, QWF refines $\bm{z}_0$ by a quaternion kind of gradient descent. Here, the gradient is calculated under   the framework of (generalized) HR calculus \cite{mandic2010quaternion,xu2015theory}, but we still follow the convention in \cite{candes2015phase} and term the algorithm as (quaternion) wirtinger flow. Given $f(\bm{z})$ with quaternion variable $\bm{z}$, it would be cumbersome to rewrite it as $f(\Re(\bm{z}),\mathcal{P}^{\ii}(\bm{z}),\mathcal{P}^{\jj}(\bm{z}),\mathcal{P}^{\kk}(\bm{z}))$ and then follow the definition (\ref{2.7}). Instead,  we apply some rules derived in \cite{xu2015theory}, specifically the product rule $\frac{\partial (fg)}{\partial \bm{q}} = f\frac{\partial g}{\partial \bm{q}}+ \frac{\partial f}{\partial {\bm{q}}}g$ for real functions $f,g$ with quaternion variable $\bm{q}$ \cite[Corollary 3.1]{xu2015theory}  and   $\frac{\partial |\bm{\alpha^*}_k\bm{z}|^2}{\partial \bm{z}} = \frac{1}{2}\bm{z^*\alpha}_k\bm{\alpha^*}_k$ \cite[Table IV]{xu2015theory}. Therefore, the quaternion derivative can be calculated as 
\textcolor{black}{
\begin{equation}
\label{derivativeexample}\begin{aligned}
      \frac{\partial f(\bm{z})}{\partial \bm{z}} &= \frac{1}{n}\sum_{k=1}^n \frac{\partial}{\partial \bm{z}} \big(|\bm{\alpha^*}_k\bm{z}|^2\cdot |\bm{\alpha^*}_k\bm{z}|^2 -2y_k\cdot |\bm{\alpha^*}_k\bm{z}|^2\big)
     \\&= \frac{2}{n}\sum_{k=1}^n  \big(|\bm{\alpha^*}_k\bm{z}|^2-y_k\big)\cdot  \frac{\partial |\bm{\alpha^*}_k\bm{z}|^2}{\partial \bm{z}} \\&= \frac{1}{n}\sum_{k=1}^n \big(|\bm{\alpha^*}_k\bm{z}|^2-|\bm{\alpha^*}_k\bm{x}|^2\big)\bm{z^*\alpha}_k\bm{\alpha^*}_k.
     \end{aligned}
\end{equation}}
 Hence, taking a suitable step size $\eta$, the   update    rule is\begin{equation}
 \label{3.17}
     \bm{z}_{t+1} = \bm{z}_t - \eta\cdot \nabla f(\bm{z}_t),
 \end{equation}
\textcolor{black}{where we let $\nabla f(\bm{z}) := \big(\frac{\partial f(\bm{z})}{\partial\bm{z}}\big)^*$ to keep notation light, i.e., \begin{equation}
     \begin{aligned}
         \label{nablala}
        \nabla f(\bm{z})   = \frac{1}{n}\sum_{k=1}^n\big(|\bm{\alpha}_k^*\bm{z}|^2-y_k\big)\bm{\alpha}_k\bm{\alpha}_k^*\bm{z}
     \end{aligned}
 \end{equation}  } Overall, we summarize the QWF update in {Algorithm \ref{alg2}}.
 \begin{algorithm}
 \caption{Quaternion Wirtinger Flow (QWF)}\label{alg2}
     \begin{algorithmic}[1]
\Statex \textbf{Input:} $(\bm{\alpha}_k,y_k)_{k=1}^n$, step size $\eta$, iteration number $T$

\item  
\textbf{for}~{$i=0,1,...,T-1$}: 

Compute
$
    \nabla f(\bm{z}_{i})
$ as in (\ref{nablala}), 
  then  update $\bm{z}_{i}$ to $\bm{z}_{i+1}$ as in (\ref{3.17}).  
  
  \noindent\textbf{end for}

  \Statex \textbf{Output:} $\bm{z}_T$
     \end{algorithmic}
 \end{algorithm}

 \subsection{Linear Convergence}
 For $\bm{z},\bm{x}\in\mathbb{Q}^d$, due to the trivial ambiguity of  right  quaternion phase factor, we characterize the distance between $\bm{z}$ and $\bm{x}$ by 
 \begin{equation}\label{distance}
     \dist(\bm{z},\bm{x}) = \min_{\mathtt{w}\in \mathbb{T}_\mathbb{Q}} \|\bm{z}-\bm{x}\mathtt{w}\|.
 \end{equation}
 Define the phase of nonzero $\mathtt{w}\in\mathbb{Q}$ to be $\sign(\mathtt{w}) = \frac{\mathtt{w}}{|\mathtt{w}|}$, and let $\sign(0)=1$, some algebra  shows that the minimum of (\ref{distance}) is attained at $\mathtt{w} = \sign(\bm{x^*z})$, and hence
 $\dist(\bm{z},\bm{x}) = \|\bm{z}-\bm{x}\sign(\bm{x^*z})\|.$ In our setting, $\bm{x}$ is the fixed underlying signal, hence we  write $\dist(\bm{z},\bm{x}) = \|\bm{z} - \bm{x}\cdot\phi(\bm{z})\|$ with $
     \phi(\bm{z}):= \sign(\bm{x^*z}) 
$  to denote the reconstruction error of $\bm{z}$. Naturally, a small neighborhood of $\bm{x}$ should be given as $E_\epsilon(\bm{x}) = \{\bm{z}\in\mathbb{Q}^d:\dist(\bm{z},\bm{x})\leq \epsilon\}.$ We   present our first main result that guarantees the linear convergence of QWF.
 \begin{theorem}
 \label{theorem2}
 We consider a fixed signal $\bm{x}$ satisfying $\|\bm{x}\|=1$, a measurement matrix $\bm{A}\sim \mathcal{N}_\mathbb{Q}^{n\times d}$, and the observations $y_k=|\bm{\alpha}_k^*\bm{x}|^2$ where $\bm{\alpha}_k$ is the $k$-th row of $\bm{A}$. Suppose that we run Algorithm \ref{alg2} with the step size $\eta$ in (\ref{3.17}) satisfying $\eta=O(\frac{1}{d})$.   
Then under the sample size $n\geq C_1d\log n$ for some absolute constant $C_1$, with probability at least $1-C_2n^{-9}-C_3n\exp(-C_4d)$, the sequence $\{\bm{z}_t\}$ produced by   QWF satisfies 
 \begin{equation}
 \label{3.199}
      \dist^2(\bm{z}_{t+1},\bm{x}) \leq \big(1- \frac{c_1}{d}\big) \dist^2(\bm{z}_t,\bm{x})
 \end{equation}
 for some $c_1$. 
 \end{theorem}
 \begin{rem}
     We assume $\|\bm{x}\|=1$ to facilitate theoretical analysis with no loss of generality. To be adaptive to an unknown signal norm, as in \cite{candes2015phase} we suggest a step size $\eta = \frac{\eta_1}{\|\bm{z}_0\|^2}$ where $\bm{z}_0$ is the spectral initialization from Algorithm \ref{alg1}. In this case, the linear convergence still holds as long as $\eta_1=O(\frac{1}{d})$.  
 \end{rem}
The proof of Theorem 2 can be divided into several ingredients below, specifically Lemmas \ref{lemma4}--\ref{lemadd}, and then we will arrive at the desired linear convergence in the end of this section. 
 The theoretical analysis will be provided in a reverse order. We first show that as long as $\bm{z}_0\in E_\epsilon(\bm{x})$ for some sufficiently small $\epsilon$, the sequence produced by QWF update (\ref{3.17}) linearly converges to $\bm{x}$. Then, we complete the proof by showing   $\bm{z}_0\in E_\epsilon(\bm{x})$ holds with high probability.
To analyze the behaviour of $\{\bm{z}_t\}$ in $E_\epsilon(\bm{x})$, we define several conditions to  characterize the landscape of $f(\bm{z})$ when $\bm{z}\in E_\epsilon(\bm{x})$.

\begin{con}
\label{con1}
 {\rm 	 (Regularity Condition)}  The regularity condition holds with positive parameters $\tau,\beta,\epsilon$, abbreviated as $\mathrm{RC}(\tau,\beta,\epsilon)$, if \begin{equation}
     \begin{aligned}
     \Re\big<&\nabla f(\bm{z}),\bm{z}-\bm{x}\cdot\phi(\bm{z})\big> \geq \frac{1}{\tau} \dist^2(\bm{z},\bm{x}) \\&~~~~~~~+\frac{1}{\beta} \|\nabla f(\bm{z})\|^2,~\forall~\bm{z}\in E_\epsilon(\bm{x}).
     \end{aligned}
\label{3.19} \end{equation}
\end{con}
 With sufficiently small step size, linear convergence of $\{\bm{z}_t\}$ can be implied by RC$(\tau,\beta,\epsilon)$ of $f(\bm{z})$. This observation bears resemblance  to a classical result in convex optimization (see \cite[Theorem 2.1.15]{nesterov2003introductory}), but its proof requires proper modification (see  \cite[Lemma 7.10]{candes2015phase}). Fortunately, it remains true in   the quaternion setting without essential  technical changes.
  \begin{lem}
\label{lemma4}
Under Condition \ref{con1}, if $\bm{z}_0 \in E_\epsilon(\bm{x})$ and the step size $0<\eta \leq \frac{2}{\beta}$, then the sequence $\{\bm{z}_t\}$ produced by (\ref{3.17}) satisfies \begin{equation}
\label{B.1}
    \dist^2(\bm{z}_{t+1},\bm{x}) \leq \Big(1- \frac{2\eta}{\tau}\Big) \dist^2(\bm{z}_t,\bm{x}).
\end{equation}
\end{lem}
 \begin{proof}
     The proof can be found in supplementary material.
 \end{proof}
  
  Therefore, it suffices to establish the regularity condition,  and similar to \cite{candes2015phase}  we   divide it into two properties.
  
  \begin{con}\label{con2}
  {\rm 	(Local Curvature Condition)} The local curvature condition holds with positive parameters $ \tau,\beta,\epsilon$, abbreviated as  $\mathrm{LCC}(\tau,\beta,\epsilon)$, if  \begin{equation}
  \label{3.20}
      \begin{aligned}
      \Re\big<&\nabla f(\bm{z}),\bm{z}-\bm{x}\phi(\bm{z})\big> \geq \frac{1}{\tau} \dist^2(\bm{z},\bm{x})\\&+\frac{1}{\beta} \frac{1}{n}\sum_{k=1}^n |\bm{\alpha^*}_k(\bm{z}-\bm{x}\phi(\bm{z}))|^4,~\forall~\bm{z}\in E_\epsilon(\bm{x}).
      \end{aligned}
  \end{equation}
  \end{con}
  \begin{con}\label{con3}
  {\rm 	(Local Smoothness Condition)} The Local Smoothness Condition holds with positive parameters $\tau,\beta,\epsilon$, abbreviated as $\mathrm{LSC}(\tau,\beta,\epsilon)$, if\begin{equation}
     \begin{aligned}
      \|\nabla f(&\bm{z})\|^2 \leq \frac{1}{\tau}\dist^2(\bm{z},\bm{x}) \\&+\frac{1}{\beta} \sum_{k=1}^n \frac{1}{n}|\bm{\alpha^*}_k(\bm{z}-\bm{x}\phi(\bm{z}))|^4,~\forall~\bm{z}\in E_\epsilon(\bm{x}).
     \end{aligned}
  \end{equation} 
  \end{con}
   \subsubsection{Local Curvature Condition}
  \begin{lem}
  \label{lem2}{\rm 	(Proving LCC)}
  Assume $\bm{x}\in\mathbb{Q}^d$ is a fixed underlying signal. Given   $\epsilon\in[0,1]$ and sufficiently small $\frac{1}{\tau}$, $\frac{1}{\beta}$, when $n= \Omega(d\log n)$ for sufficiently large hidden constant, with probability at least $1-32n^{-9} -C_1n\exp(-C_2d)$, $\mathrm{LCC}(\tau,\beta,\epsilon)$ in Condition \ref{con2} is satisfied. 
  \end{lem}
  \noindent
  {\it Proof.} We aim to show (\ref{3.20}) for some $\tau,\beta$ specified later.  We   define $\bm{h}_0 = \bm{z}\overline{\phi(\bm{z})} - \bm{x}$, then $\bm{z}\in E_\epsilon(\bm{x})$ translates into $\|\bm{h}_0\|\leq \epsilon$, also $\phi(\bm{z}) = \sign(\bm{x^*z})$ implies $\Im (\bm{h^*}_0\bm{x}) = 0$.  Then we deal with (\ref{3.20}) by using $\nabla f(\bm{z}) = \frac{1}{n}\sum_{k=1}^n \big(|\bm{\alpha^*}_k\bm{z}|^2-|\bm{\alpha^*}_k\bm{x}|^2\big)\bm{\alpha}_k\bm{\alpha^*}_k\bm{z}$ and $\bm{z}=(\bm{h}_0+\bm{x}) \phi(\bm{z})$, it gives a sufficient condition for (\ref{3.20}) as $ \forall~\|\bm{h}_0\|\leq \epsilon,\Im(\bm{h^*}_0\bm{x})=0$,
  \begin{equation}
      \begin{aligned}
      \nonumber
      &\frac{2}{n}\sum_{k=1}^n \big[\Re(\bm{x^*\alpha}_k\bm{\alpha^*}_k\bm{h}_0)\big]^2 + \frac{3}{n}\sum_{k=1}^n |\bm{\alpha^*}_k \bm{h}_0| ^2 \Re (\bm{x^*\alpha}_k\bm{\alpha^*}_k\bm{h}_0) \\&~~~~~~~~~~~~+\big(1-\frac{1}{\beta}\big) \frac{1}{n}\sum_{k=1} ^n |\bm{\alpha^*}_k\bm{h}_0|^4 \geq \frac{1}{\tau} \|\bm{h}_0\|^2.
      \end{aligned}
 \end{equation}
  We only need to consider nonzero $\bm{h}_0$ and we further let $\bm{h}_0 = s \cdot \bm{h}$ with $s = \|\bm{h}_0\|\in [0,\epsilon]$, $\|\bm{h}\| = 1$, $\Im (\bm{h^*x}) = 0$. Hence, the above sufficient condition can be implied by $\forall~\|\bm{h} \|=1,s\in [0,\epsilon],\Im(\bm{h^*} \bm{x})=0$,
  \begin{equation}
      \begin{aligned}
      \label{3.23}
      &\frac{2}{n}\sum_{k=1}^n \big[\Re(\bm{x^*\alpha}_k\bm{\alpha^*}_k\bm{h} ) \big]^2 + \frac{3s}{n}\sum_{k=1}^n |\bm{\alpha^*}_k \bm{h} | ^2 \Re (\bm{x^*\alpha}_k\bm{\alpha^*}_k\bm{h} ) \\&~~~~~~~~~~+\big(1-\frac{1}{\beta}\big) \frac{s^2}{n}\sum_{k=1} ^n |\bm{\alpha^*}_k\bm{h} |^4 \geq \frac{1}{\tau}.
      \end{aligned}
  \end{equation}
  We define $t_\beta : = \frac{9}{8(1-\frac{1}{\beta})}$, completing the square, (\ref{3.23}) is equal to   $\forall~\|\bm{h}\|=1$, $s\in [0,\epsilon]$, $\Im(\bm{h^*x}) = 0$, \begin{equation}
  \label{3.24}
      \begin{aligned}
      &\frac{1}{n}\sum_{k=1}^n \Big(\sqrt{2t_\beta} \Re(\bm{x^*\alpha}_k\bm{\alpha^*}_k\bm{h})+\frac{3s}{\sqrt{8t_\beta}}|\bm{\alpha^*}_k\bm{h}|^2\Big)^2 \\&\geq \frac{1}{\tau}+\frac{2(t_\beta -1)}{n}\sum_{k=1}^n \big[\Re(\bm{x^*\alpha}_k\bm{\alpha^*}_k\bm{h})\big]^2.
      \end{aligned}
  \end{equation}
  We define $$Y_k(\bm{h},s) = \Big(\sqrt{2t_\beta} \Re(\bm{x^*\alpha}_k\bm{\alpha^*}_k\bm{h})+\frac{3s}{\sqrt{8t_\beta}}|\bm{\alpha^*}_k\bm{h}|^2\Big)^2 ,$$  $$\big<Y_k(\bm{h},s)\big>= \frac{1}{n}\sum_{k=1}^n Y_k(\bm{h},s).$$ By Lemma \ref{lem3}(c) we have $$\mathbbm{E} \big[\Re(\bm{x^*\alpha}_k\bm{\alpha^*}_k\bm{h})\big]^2 = \frac{1}{4}+\frac{5}{4}\big[\Re(\bm{x^*h})\big]^2,$$ and we need to work out the concentration of $\frac{1}{n}\sum_{k=1}^n \big[\Re(\bm{x^*\alpha}_k\bm{\alpha^*}_k\bm{h})\big]^2$ around its mean. Note that $\Re(\bm{x^*\alpha}_k\bm{\alpha^*}_k\bm{h})$ is just the $(1,1)$-th entry of $$\mathcal{T}(\bm{x^*\alpha}_k\bm{\alpha^*}_k\bm{h})= \mathcal{T}(\bm{x})^\top \mathcal{T}(\bm{\alpha}_k)\mathcal{T}(\bm{\alpha}_k)^\top \mathcal{T}(\bm{h}),$$ we have \begin{equation}
      \begin{aligned}
          \nonumber
          \Re(\bm{x^*\alpha}_k\bm{\alpha^*}_k\bm{h}) &= \mathcal{T}_1(\bm{\alpha}_k)^\top \mathcal{T}(\bm{x})\mathcal{T}(\bm{h})^\top \mathcal{T}_1(\bm{\alpha}_k) \\&= \sum_{i=1}^4 \mathcal{T}_1(\bm{\alpha}_k)^\top \mathcal{T}_i(\bm{x})\mathcal{T}_i(\bm{h})^\top \mathcal{T}_1(\bm{\alpha}_k),
      \end{aligned}
  \end{equation} and hence \begin{equation}
      \begin{aligned}
      \label{3.25}
       &\big[\Re(\bm{x^*\alpha}_k\bm{\alpha^*}_k\bm{h})\big]^2 = \sum_{i=1}^4\sum_{j=1}^4 \mathcal{T}_i(\bm{h})^\top \Big[\Big(\mathcal{T}_1(\bm{\alpha}_k)^\top\mathcal{T}_i(\bm{x})\\& ~~~~\mathcal{T}_1(\bm{\alpha}_k)^\top\mathcal{T}_j(\bm{x})\Big)\cdot\mathcal{T}_1(\bm{\alpha}_k)\mathcal{T}_1(\bm{\alpha}_k)^\top\Big]\mathcal{T}_j(\bm{h}).
      \end{aligned}
  \end{equation} 
  Letting $$Z_k(i,j) = \big(\mathcal{T}_1(\bm{\alpha}_k)^\top\mathcal{T}_i(\bm{x})\mathcal{T}_1(\bm{\alpha}_k)^\top\mathcal{T}_j(\bm{x})\big)\cdot\mathcal{T}_1(\bm{\alpha}_k)\mathcal{T}_1(\bm{\alpha}_k)^\top,$$ some algebra gives 
  \begin{equation}
      \begin{aligned}
      \label{3.26}
       &\Big|\frac{1}{n}\sum_{k=1}^n \big[\Re(\bm{x^*\alpha}_k\bm{\alpha^*}_k\bm{h})\big]^2 - \mathbbm{E} \big[\Re(\bm{x^*\alpha}_k\bm{\alpha^*}_k\bm{h})\big]^2\Big| 
       \\&~~~\leq \sum_{i=1}^4\sum_{j=1}^4 \Big\|\frac{1}{n}\sum_{k=1}^nZ_k(i,j)-\mathbbm{E}Z_k(i,j)\Big\| 
      \end{aligned}
  \end{equation}
  Note that entries of $\mathcal{T}_1(\bm{\alpha}_k)$ are independent copies of $\frac{1}{2}\mathcal{N}(0,1)$,   by rotational invariance, without changing distribution we can assume $\mathcal{T}_i(\bm{x}) = \mathcal{T}_j(\bm{x}) = \bm{e}_1$ if $i=j$, or $\mathcal{T}_i(\bm{x}) = \bm{e}_1$, $\mathcal{T}_j(\bm{x})=\bm{e}_2$ if $i\neq j$. Hence, we can invoke Lemma \ref{lem6} and obtain that when $n=\Omega(\delta^{-2}d\log n)$, the right-hand side of (\ref{3.26}) is bounded by $\delta$ with   probability at least $1-32n^{-9}-32\exp(-\Omega(d))$. This gives rise to   $$\frac{1}{n}\sum_{k=1}^n \big[\Re(\bm{x^*\alpha}_k\bm{\alpha^*}_k\bm{h})\big]^2\leq \frac{1}{4}+\frac{5}{4}\big[\Re(\bm{x^*h})\big]^2+\delta,$$ hence for showing (\ref{3.24}), w.h.p it suffices to show 
  \begin{equation}
  \label{3.27}
      \begin{aligned}
      &\big<Y_k(\bm{h},s)\big>\geq \frac{1}{\tau}+2(t_\beta-1)\big(\delta+\frac{1}{4}+\frac{5}{4}\big[\Re(\bm{x^*h})\big]^2\big),\\&~~~~~~~~~~~~~\forall~\|\bm{h}\|=1,\Im(\bm{h^*x}) = 0,s\in [0,\epsilon].
      \end{aligned}
  \end{equation}
  Our strategy is to first consider fixed $(\bm{h},s)$ and then apply a covering argument. Note that \begin{equation}
      \begin{aligned}
      \label{3.28}
       & \mu_k(\bm{h},s):= \mathbbm{E}Y_k(\bm{h},s)=\frac{9s}{2} \Re(\bm{h^*x}) +\frac{27s^2}{16t_\beta} +\frac{t_\beta}{2}\\& +\frac{5t_\beta}{2}\big[\Re(\bm{x^*h})\big]^2\leq \frac{9}{2}+ 3t_\beta + \frac{27}{16t_\beta} \leq C_1,
      \end{aligned}
  \end{equation}
  where we use Lemma \ref{lem3}(b)-(d). Assuming $s\leq \epsilon\leq 1$,   the last inequality follows as long as $t_\beta = \frac{9}{8(1-\frac{1}{\beta})} = O(1)$. We define $X_k(\bm{h},s) = \mu_k(\bm{h},s) - Y_k(\bm{h},s)$, then $\mathbbm{E}X_k(\bm{h},s)=0$, $X_k(\bm{h},s)\leq \mu_k(\bm{h},s)\leq C_1$. Assuming $t_\beta = O(1)$, we further estimate the  variance 
  \begin{equation}
      \begin{aligned}
      \nonumber
       &\mathbbm{E}\big[X_k(\bm{h},s)^2\big] \leq \mathbbm{E}\big[Y_k(\bm{h},s)^2\big] \\& = \mathbbm{E}\Big(\sqrt{2t_\beta} \Re(\bm{x^*\alpha}_k\bm{\alpha^*}_k\bm{h})+\frac{3s}{\sqrt{8t_\beta}}|\bm{\alpha^*}_k\bm{h}|^2\Big)^4\\
       &\leq 16\Big(4t_\beta^2 \mathbbm{E}\big[\Re(\bm{x^*\alpha}_k\bm{\alpha^*}_k\bm{h})\big]^4+ \frac{81}{64t_\beta^2}\mathbbm{E} |\bm{\alpha^*}_k\bm{h}|^8\Big) \leq C_2.
      \end{aligned}
  \end{equation}
  Now we can invoke   \cite[Lemma 7.13]{candes2015phase} (or the original derivation \cite{bentkus2003inequality}) and obtain $\forall~t>0$ \begin{equation}
  \label{3.30}
      \mathbbm{P}\Big(\mu_k(\bm{h},s)- \big< Y_k(\bm{h},s)\big>\geq t\Big) \leq \exp(-C_3nt^2).
  \end{equation}
  On the other hand, we can use the first line of (\ref{3.28}), then for fixed $\bm{h},s$, (\ref{3.27}) becomes 
  \begin{equation}
  \label{3.31}
      \begin{aligned}
      &\big<Y_k(\bm{h},s)\big> - \mu_k(\bm{h},s) \geq \frac{1}{\tau} +(2t_\beta -2)\delta -\frac{1}{2} \\&~~~~~~~~~~~~~-\frac{5}{2}\big[\Re(\bm{x^*h})\big]^2-\frac{9s}{2}\Re(\bm{h^*x}) -\frac{27s^2}{16t_\beta}.
      \end{aligned}
  \end{equation}
Setting $\epsilon$ (hence $|s|$), $\delta$, and $\frac{1}{\tau}$ to be sufficiently small, $t_\beta = O(1)$, then the right-hand side of (\ref{3.31}) can be upper bounded by $-\frac{1}{4}$, hence (\ref{3.31}) is implied by $ \mu_k(\bm{h},s)-\big<Y_k(\bm{h},s)\big>  \leq\frac{1}{4}$. So it remains to show $$\sup_{\bm{h},s} \big(\mu_k(\bm{h},s) - \big<Y_k(\bm{h},s)\big>\big)\leq \frac{1}{4},$$ where the supremum is taken over $s\in [0,\epsilon]$, $\|\bm{h}\|=1$.
For fixed $\bm{h},s$, by (\ref{3.30}), $$\mathbbm{P}\big(\mu_k(\bm{h},s) - \big<Y_k(\bm{h},s)\big>\geq \frac{1}{8}\big)\leq \exp(-\frac{C_3}{64}n).$$  We then apply a covering argument, specifically we construct a $\delta_s$-net $\mathcal{N}_s$ of $[0,\epsilon]$, a $\delta_{\bm{h}}$-net $\mathcal{N}_{\bm{h}}$ of $\{\bm{h}:\|\bm{h}\|=1\}$, then a union bound delivers \begin{equation}
\label{3.32}
   \begin{aligned}
    &\mathbbm{P}\Big(\max_{s\in \mathcal{N}_s}\max_{\bm{h}\in\mathcal{N}_{\bm{h}}} \big[\mu_k(\bm{h},s) - \big<Y_k(\bm{h},s)\big>\big]\geq \frac{1}{8}\Big)\\&~~~~~~~~~~~~~~~~~~~\leq |\mathcal{N}_s||\mathcal{N}_{\bm{h}}| \exp(-\frac{C_3}{64}n).
   \end{aligned}
\end{equation} 
We can assume \begin{equation}
    \begin{aligned}
        \nonumber
        &\sup_{s\in [0,\epsilon]}\sup_{\|\bm{h}\|=1} \big[\mu_k(\bm{h},s)-\big<Y_k(\bm{h},s)\big>\big] \\&= \mu_k(\bm{h}_0,s_0) - \big<Y_k(\bm{h}_0,s_0)\big>
    \end{aligned}
\end{equation} for some $s_0\in [0,\epsilon]$, $\|\bm{h}_0\|=1$. We can further pick $\bm{h}_1\in\mathcal{N}_{\bm{h}}$, $s_1\in \mathcal{N}_s$ such that $\|\bm{h}_1-\bm{h}_0\|\leq \delta_{\bm{h}}$, $|s_1-s_0|\leq \delta_s$. Similar to \cite{candes2015phase}, for some $C_4$, $\max_{k\in [n]}\|\bm{\alpha}_k\|\leq C_4\sqrt{d}$ holds with probability at least $1-n\exp(-C_5d)$ (One may also see this by a direction application of Theorem 3.1.1, \cite{vershynin2018high}). We proceed on this assumption, and start from \begin{equation}
    \begin{aligned}
    \label{3.33}
     &\big|\big[\mu_k(\bm{h}_0, s_0)-\big<Y_k(\bm{h}_0,s_0)\big>\big] - \big[\mu_k(\bm{h}_1, s_1)-\big<Y_k(\bm{h}_1,s_1)\big>\big]\big|\\
     &\leq \big|\mu_k(\bm{h}_0, s_0)-\mu_k(\bm{h}_1, s_1)\big|+\\&~~~~~~~~~~\big|\frac{1}{n}\sum_{k=1}^n\big(Y_k(\bm{h}_0,s_0)-Y_k(\bm{h}_1,s_1)\big)\big| :=R_1+R_2.
    \end{aligned}
\end{equation}
We use the first line in (\ref{3.28}), it is direct to show $R_1\leq C_6 (\delta_{\bm{h}}+\delta_s)$. For estimate of $R_2$, we let $P_k(\bm{h},s) = \sqrt{2t_\beta} \Re(\bm{x^*\alpha}_k\bm{\alpha^*}_k\bm{h}) + \frac{3s}{\sqrt{8t_\beta}}|\bm{\alpha^*}_k\bm{h}|^2$, and note that $Y_k(\bm{h},s)=P_k(\bm{h},s)^2$, $P_k(\bm{h},s)\lesssim d$ due to $t_\beta = O(1)$, $\max_k\|\bm{\alpha}_k\|\leq C_4\sqrt{d}$. Then $R_2=$
\begin{equation}
    \begin{aligned}
    \nonumber
     &  \big|\frac{1}{n}\sum_{k=1}^n (P_k(\bm{h}_0,s_0)-P_k(\bm{h}_1,s_1))\cdot(P_k(\bm{h}_0,s_0)+P_k(\bm{h}_1,s_1))\big|\\
     &\lesssim\frac{d}{n} \sum_{k=1}^n \Big(|\Re(\bm{x^*\alpha}_k\bm{\alpha^*}_k(\bm{h}_0-\bm{h}_1))| +|s_0-s_1||\bm{\alpha^*}_k\bm{h}_0|^2 \\& +s_1 \big(|\bm{\alpha^*}_k\bm{h}_0|^2-|\bm{\alpha^*}_k\bm{h}_1|^2\big)\Big)\lesssim d^2(\delta_{\bm{h}}+\delta_s).
    \end{aligned}
\end{equation}
Thus, we can take $\delta_{\bm{h}},\delta_s= \frac{1}{C_7d^2}$ with sufficiently large $C_7$ so that $R_1+R_2 \leq \frac{1}{8}$ holds. In this case we can assume $$|\mathcal{N}_s|\leq \frac{\epsilon}{1/C_7d^2} \leq C_7d^2, |\mathcal{N}_{\bm{h}}|\leq (1+2C_7d^2)^{4d}.$$ Plug these into (\ref{3.32}), when $n = \Omega(d\log d)$ for sufficiently large hidden constant, with probability at least $1-\exp \big(-\frac{C_3n}{128}\big)$, $$\max_{s\in\mathcal{N}_s}\max_{\bm{h}\in \mathcal{N}_{\bm{h}}} \big[\mu_k(\bm{h},s)-\big<Y_k(\bm{h},s)\big>\big]\leq \frac{1}{8},$$ which together with (\ref{3.33}) yields 
\begin{equation}
\nonumber
    \begin{aligned}
    &\mu_k(\bm{h}_0,s_0)- \big<Y_k(\bm{h}_0,s_0)\big> \leq R_1+R_2+\\& \max_{s\in\mathcal{N}_s}\max_{\bm{h}\in \mathcal{N}_{\bm{h}}} \big[\mu_k(\bm{h},s)-\big<Y_k(\bm{h},s)\big>\big]\leq \frac{1}{4}.
    \end{aligned}
\end{equation}
Recall that the only additional scaling we assume in the proof is $t_\beta = O(1)$, while this can be guaranteed by sufficiently small $\frac{1}{\beta}$. Hence, the proof is concluded. \hfill $\square$
 \begin{rem}
 \label{rem4}
Compared to the proof in complex case (Section VII of \cite{candes2015phase}), we need new machinery to deal with some technical issues. For instance, the concentration of  $\widehat{R}:=\frac{1}{n}\sum_{k=1}^n \big[\Re (\bm{x^*\alpha}_k\bm{\alpha^*}_k\bm{h})\big]^2$ in (\ref{3.26}). Specifically, \cite{candes2015phase} used the concentration of the Hessian matrix $\nabla^2f(\bm{x})$ to govern the whole proof, which could yield the concentration of $\widehat{R}$ by a clever observation $\widehat{R} = \frac{1}{4}\bm{\widehat{h}^*}\nabla ^2f(\bm{x})\bm{\widehat{h}}$ where $\bm{\widehat{h}^*} = [\overline{\bm{h}},\bm{h}]$ (Corollary 7.5, \cite{candes2015phase}). However, this becomes infeasible in quaternion setting: Firstly,   the Hessian matrix now contains $16$ blocks and can be exhausting in calculations (see     \cite[Equation (33)]{xu2015optimization}); Perhaps more prominently, the relation between $\widehat{R}$ and $\bm{\widehat{h}^*}\nabla f(\bm{x})\bm{\widehat{h}}$ heavily relies on commutativity and hence  is likely to fail due to non-commutativity of quaternion. Instead, we   calculate $\big[\Re(\bm{x^*\alpha}_k\bm{\alpha^*}_k\bm{h})\big]^2$ via the map $\mathcal{T}(\cdot)$ (\ref{3.25}). Having reduced to the real case, we directly work on the desired concentration ingredient in Lemma \ref{lem6}. 
 \end{rem}
 
 \subsubsection{Local Smoothness Condition}
 \begin{lem}
 \label{lemaddd}
 {\rm 	(LSC)} Assume $\bm{x}\in \mathbb{Q}^d$ is the fixed underlying signal. Given $\epsilon\in [0,1]$ and $\tau = C_0$, $\beta = C_1/d$ with sufficiently small $C_0,C_1$. If $n = \Omega(d\log n)$ with sufficiently large hidden constant, with probability at least $1-C_2n^{-9}-C_3n\exp(-\Omega(d))$, $\mathrm{LSC}(\tau,\beta, \epsilon)$ in Condition \ref{con3} is satisfied.
 \end{lem}
 \noindent{\it Proof.} Writing $$\|\nabla f(\bm{z})\| = \sup_{\|\bm{u}\|=1}\Re(\bm{u}^* \nabla f(\bm{z})),$$ the desired $\mathrm{LSC}(\tau,\beta,\epsilon)$ is equivalent to ($\tau,\beta$ will be specified later) \begin{equation}
 \label{3.36}
    \begin{aligned}
     &\big|\Re(\bm{u^*}\nabla f(\bm{z}))\big|^2\leq \frac{1}{\tau}\|\bm{z}-\bm{x}\phi (\bm{z})\|^2 + \\&\frac{1}{\beta}\frac{1}{n} \sum_{k=1}^n\big|\bm{\alpha^*}_k(\bm{z}-\bm{x}\phi (\bm{z}))\big|^4,~\forall \bm{z}\in E_\epsilon(\bm{x}),\|\bm{u}\|=1.
    \end{aligned}
 \end{equation}
 We let $\bm{h}:= \bm{z}\overline{\phi(\bm{z})}-\bm{x}$, then $\|\bm{h}\|\leq \epsilon$, $\Im(\bm{h^*x})=0$, and $\bm{z} = (\bm{h}+\bm{x})\phi(\bm{z})$. Substituting $\bm{z}$ with $\bm{h}$, the right-hand side of (\ref{3.36}) becomes $$\frac{1}{\tau}\|\bm{h}\|^2+\frac{1}{\beta}\frac{1}{n} \sum_{k=1}^n |\bm{\alpha^*}_k \bm{h}|^4.$$ We further define $\bm{w} = \bm{u}\overline{\phi(\bm{z})}$, and plug in $\nabla f(\bm{z})$, some algebra gives\begin{equation}
     \begin{aligned}
     \nonumber
      &\Re(\bm{u^*}\nabla f(\bm{z})) = \frac{1}{n}\sum_{k=1}^n \Big( |\bm{\alpha^*}_k\bm{h}| ^2 \Re(\bm{w^*\alpha}_k\bm{\alpha^*}_k\bm{h})\\&+|\bm{\alpha^*}_k\bm{h}|^2 \Re(\bm{w^*\alpha}_k\bm{\alpha^*}_k\bm{x})
      + 2\Re(\bm{h^*\alpha}_k\bm{\alpha^*}_k\bm{x})\Re(\bm{w^*\alpha}_k\bm{\alpha^*}_k\bm{h}) \\&+2\Re(\bm{h^*\alpha}_k\bm{\alpha^*}_k\bm{x})\Re(\bm{w^*\alpha}_k\bm{\alpha^*}_k\bm{x})\Big).
     \end{aligned}
 \end{equation}
 Thus we can further estimate the left-hand side of (\ref{3.36}) \begin{equation}
     \begin{aligned}
     \nonumber
      & |\Re(\bm{u^*}\nabla f(\bm{z}))|^2 \leq \frac{1}{n^2} \Big| \sum_{k=1}^n \Big(|\bm{\alpha^*}_k\bm{h}|^3|\bm{\alpha^*}_k\bm{w}| \\& +3 |\bm{\alpha^*}_k\bm{h}|^2|\bm{\alpha^*}_k\bm{x}||\bm{\alpha^*}_k\bm{w}|+2|\bm{\alpha^*}_k\bm{h}| |\bm{\alpha^*}_k\bm{x}|^2|\bm{\alpha^*}_k\bm{w}|\Big)\Big|^2\\
      &\leq 3\Big(\big[\frac{1}{n}\sum_{k=1}^n|\bm{\alpha^*}_k\bm{h}|^3|\bm{\alpha^*}_k\bm{w}|\big]^2+9\big[\frac{1}{n}\sum_{k=1}^n|\bm{\alpha^*}_k\bm{h}|^2|\bm{\alpha^*}_k\bm{x}||\bm{\alpha^*}_k\bm{w}| \big]^2\\&+4\big[\frac{1}{n}\sum_{k=1}^n|\bm{\alpha^*}_k\bm{h}||\bm{\alpha^*}_k\bm{x}|^2|\bm{\alpha^*}_k\bm{w}|\big]^2\Big) := 3\Big(I_1+9I_2+4I_3\Big).
     \end{aligned}
 \end{equation}
 Similar to the proof of Lemma \ref{lem2}, we can assume $\max_{k\in[n]}\|\bm{\alpha}_k\|\leq C_1\sqrt{d}$ with probability at least $1-n\exp(-cd)$. Since $\|\bm{h}\|\leq \epsilon\leq 1$, by Cauchy-Schwarz we have \begin{equation}
     \begin{aligned}
     \label{a3.41}
      &I_1  =\big[\frac{1}{n}\sum_{k=1}^n|\bm{\alpha^*}_k\bm{h}|^3|\bm{\alpha^*}_k\bm{w}|\big]^2 \\& \lesssim d \cdot \big[\sum_{k=1}^n \big(\frac{1}{\sqrt{n}}|\bm{\alpha^*}_k\bm{h}|^2\big)\cdot\big(\frac{1}{\sqrt{n}} |\bm{\alpha^*}_k\bm{h}|\big)\big]^2\\
      &\leq d   \Big(\sum_{k=1}^n \frac{1}{n}|\bm{\alpha^*}_k\bm{h}|^4\Big)    \Big(\sum_{k=1}^n\frac{1}{n}|\bm{\alpha^*}_k\bm{h}|^2\Big) \\& \lesssim d  \Big(\sum_{k=1}^n \frac{1}{n}|\bm{\alpha^*}_k\bm{h}|^4\Big)   \big\|\frac{1}{n}\bm{A^*A}\big\| \lesssim d\Big(\sum_{k=1}^n \frac{1}{n}|\bm{\alpha^*}_k\bm{h}|^4\Big),
     \end{aligned}
 \end{equation}
 where in the last inequality we use $\|\bm{A^*A}\|\leq \|\bm{A}\|^2$ and a standard estimate for operator norm of (sub-)Gaussian matrix that holds with probability at least $1-2\exp(-n)$ (e.g., Theorem 4.4.5, \cite{vershynin2018high}). We similarly use Cauchy-Schwarz to deal with $I_2$, it yields \begin{equation}
     \begin{aligned}
     \nonumber
      &I_2 = \Big[\sum_{k=1}^n\big(\frac{1}{\sqrt{n}}|\bm{\alpha^*}_k\bm{h}|^2\big) \big(\frac{1}{\sqrt{n}}|\bm{\alpha^*}_k\bm{x}||\bm{\alpha^*}_k\bm{w}|\big)\Big]^2 \\& \leq \Big(\frac{1}{n}\sum_{k=1}^n |\bm{\alpha^*}_k\bm{h}|^4\Big) \Big(\frac{1}{n}\sum_{k=1}^n |\bm{\alpha^*}_k\bm{x}|^2|\bm{\alpha^*}_k\bm{w}|^2\Big) \\
      &\leq \Big(\frac{1}{n}\sum_{k=1}^n |\bm{\alpha^*}_k\bm{h}|^4\Big)\cdot \big\|\frac{1}{n}\sum_{k=1}^n|\bm{\alpha^*}_k\bm{x}|^2 \bm{\alpha}_k\bm{\alpha^*}_k\big\|\lesssim \frac{1}{n}\sum_{k=1}^n |\bm{\alpha^*}_k\bm{h}|^4.
     \end{aligned}
 \end{equation}
 Note that in the last inequality, we can assume $\bm{x}=\bm{e}_1$ by rotational invariance, then write $\bm{\alpha}_k = [\mathtt{\alpha}_{ki}]$ and calculate
 \begin{equation}
     \label{add343}
     \begin{aligned}
      &\big\|\frac{1}{n}\sum_{k=1}^n |\alpha_{k1}|^2 \bm{\alpha}_k\bm{\alpha}_k^*\big\|=\big\|\frac{1}{n}\sum_{k=1}^n |\alpha_{k1}|^2 \mathcal{T}(\bm{\alpha}_k\bm{\alpha}_k^*)\big\|  \\&\leq \sum_{i=1}^4 \big\|\frac{1}{n}\sum_{k=1}^n |\alpha_{k1}|^2 \mathcal{T}_i(\bm{\alpha}_k)\mathcal{T}_i(\bm{\alpha}_k)^\top\big\| \\&\leq \sum_{i=1}^4 \big\|\frac{1}{n}\sum_{k=1}^n |\Re(\alpha_{k1})|^2 \mathcal{T}_i(\bm{\alpha}_k)\mathcal{T}_i(\bm{\alpha}_k)^\top\big\|\\
      &~~~+\sum_{i=1}^4\sum_{\vartheta = \ii,\jj,\kk} \big\|\frac{1}{n}\sum_{k=1}^n |\mathcal{P}^\vartheta(\alpha_{k1})|^2 \mathcal{T}_i(\bm{\alpha}_k)\mathcal{T}_i(\bm{\alpha}_k)^\top\big\|.
     \end{aligned}
 \end{equation} Note that $\Re(\alpha_{k1}),\mathcal{P}^\vartheta(\alpha_{k1})$ is just one entry of $\mathcal{T}_i(\bm{\alpha}_k)$, we can invoke Lemma \ref{lem6}  to establish the concentration of each summand in (\ref{add343}) around its mean with $\delta =1$, while evidently for  the mean of each summand is $O(1)$. This leads to $\big\| \frac{1}{n}\sum_{k=1}^n |\bm{\alpha^*}_k\bm{x}|^2\bm{\alpha}_k\bm{\alpha^*}_k\big\| = O(1)$ with probability at least $1-C_2n^{-9}-C_3\exp(-\Omega(d))$. We use this again to deal with $I_3$  \begin{equation}
     \begin{aligned}\nonumber
      &I_3 = \Big[\sum_{k=1}^n \big(\frac{1}{\sqrt{n}}|\bm{\alpha^*}_k\bm{h}||\bm{\alpha^*}_k\bm{x}|\big)\cdot\big(\frac{1}{\sqrt{n}}|\bm{\alpha^*}_k\bm{w}||\bm{\alpha^*}_k\bm{x}|\big)\Big]^2 \\
      &\leq \Big(\frac{1}{n}\sum_{k=1}^n |\bm{\alpha^*}_k\bm{x}|^2|\bm{\alpha^*}_k\bm{h}|^2 \Big)\cdot \Big(\frac{1}{n}\sum_{k=1}^n |\bm{\alpha^*}_k\bm{x}|^2|\bm{\alpha^*}_k\bm{w}|^2 \Big)\lesssim \|\bm{h}\|^2.
     \end{aligned}
 \end{equation}
 Putting pieces together, we have shown that for all $\|\bm{w}\|=1$, $\|\bm{h}\|\leq \epsilon$, $|\Re(\bm{u^*}\nabla f(\bm{z}))|^2 \leq C\|\bm{h}\|^2 + Cd\big(\frac{1}{n}\sum_{k=1}^n |\bm{\alpha^*}_k\bm{h}|^4\big)$. Thus, $\mathrm{LSC}(\tau, \beta, \epsilon)$ holds with sufficiently small $\tau$, and $\beta = \frac{c_1}{d}$ with sufficiently small $c_1$. The result follows. \hfill $\square$

  \subsubsection{Spectral Initialization}
  Now it is clear that, for some sufficiently small $\frac{1}{\beta_1},\beta_2$,  $\mathrm{LCC}(\beta_1,\beta_1,\epsilon)$ and $\mathrm{LSC}(\beta_2,\beta_2/d,\epsilon)$ hold simultaneously, which directly leads to $\mathrm{RSC}(2\beta_1d,\frac{\beta_1}{\beta_2}d,\epsilon)$. By Lemma \ref{lemma4},    the QWF sequence with $\eta \leq \frac{2}{d}$     linearly converges to $\bm{x}$ if some iteration point is sufficiently close to $\bm{x}$. Thus, the proof can be concluded by showing $\bm{z}_0 \in E_\epsilon(\bm{x})$, which is presented in Lemma \ref{lemadd}.  
  \begin{lem}
  \label{lemadd}
  Assume $\bm{x}$ is the fixed underlying signal. Given $\delta \in (0,1]$. If $n \geq C_0 \delta^{-2}d\log n$ for sufficiently large hidden constant, then with probability at least $1-C_1n^{-9}-C_2\exp(-C_3d)$, $\bm{z}_0\in E_{2\delta}(\bm{x})$.
  \end{lem}
  \begin{proof}
     The proof can be found in supplementary material.
 \end{proof}
  In Lemma \ref{lemadd} we take $\delta = \frac{1}{16}$, then under the assumptions of Theorem \ref{theorem2} $\bm{z}_0 \in E_{1/8}(\bm{x})$ with
		high probability. Then, applying Lemma \ref{lem2}, \ref{lemaddd} shows $\mathrm{RSC}(c,dc,\frac{1}{8})$ (where $c$ is sufficiently large). Thus, if  $\eta = O(\frac{1}{d})$, Lemma \ref{lemma4} delivers the linear convergence claimed in Theorem \ref{theorem2}.

  


\section{Pure Quaternion Wirtinger Flow}\label{sec4}
Recall that $\mathbb{Q}_p$ is the set of pure quaternions, and naturally, $\mathbb{Q}_p^d$ represents the space of $d$-dimensional pure quaternion signals.
This section is intended to propose a variant of QWF called pure quaternion Wirtinger flow (PQWF) that can effectively      utilize the priori of $\bm{x}\in \mathbb{Q}_p^d$. This is motivated, for example, by  quaternion methods in color image processing where the color channels  are   encoded in three imaginary components, and hence the desired signal is pure quaternion   \cite{jia2019robust,chen2019low}. While many works choose to remove the real part of the   quaternion signal after the reconstruction (e.g., \cite{jia2019robust,chen2019low,miao2021color}), it is obviously more sensible to incorporate the pure quaternion priori into  the recovery procedure and gain some benefits \cite{chen2022color,song2021low}.

For $\bm{a}\in \mathbb{Q}^d$ we define the real counterpart $\mathcal{V}(\bm{a}):= [\Re(\bm{a}),\mathcal{P}^{\ii}(\bm{a}),\mathcal{P}^{\jj}(\bm{a}),\mathcal{P}^{\kk}(\bm{a})] \in \mathbb{R}^{d\times 4}$. For instance,  we have $\mathcal{V}(\bm{x}) = [\bm{0},\mathcal{P}^{\ii}(\bm{x}),\mathcal{P}^{\jj}(\bm{x}),\mathcal{P}^{\kk}(\bm{x})]$ if $\bm{x}\in \mathbb{Q}_p^d$. The next lemma shows that, if the pure quaternion signal $\bm{x}$ satisfies $\rank\big(\mathcal{V}(\bm{x})\big) = 3$,
then the trivial ambiguity in phase retrieval reduces to a sign. 
\begin{lem}
\label{lem66}
Assume $\bm{x}\in \mathbb{Q}^d_p$. In the phase-less measurement setting described in Theorem \ref{theorem1}, all   $\bm{x}$ satisfying $\rank\big(\mathcal{V}(\bm{x})\big)=3$ can be reconstructed from $\{|\bm{\alpha^*}_k\bm{x}|^2:k\in [n]\}$ up to a sign $\pm 1$. 
\end{lem}

\noindent{\it Proof.} By Theorem \ref{theorem1}, one can exactly   reconstruct $\mathcal{A}_{\bm{x}}:= \{\bm{x}   \mathtt{q}:\mathtt{q}\in\mathbb{T}_\mathbb{Q}\}$. Due to the assumption of $\bm{x}\in \mathbb{Q}^d_p$, we choose $\bm{x}_0\in\mathcal{A}_{\bm{x}}$ and can further pick   $\mathtt{q}=q_0+q_1\ii+q_2\jj+q_3\kk\in\mathbb{T}_\mathbb{Q}$ such that $\Re(\bm{x}_0\mathtt{q})= (\Re \bm{x}_0 )q_0 - (\mathcal{P}_{\ii} \bm{x}_0 )q_1  - (\mathcal{P}_{\jj}\bm{x}_0 )q_2  - (\mathcal{P}_{\kk} \bm{x}_0 )q_3  =0$. It is not hard to verify that $\rank\big(\mathcal{V}(\bm{x}_0)\big) =  \rank\big(\mathcal{V}(\bm{x} )\big)= 3$, for example, one can identify $\mathcal{V}(\bm{x}_0)$ with the first row of $\mathcal{T}(\bm{x}_0)$ up to permutation and then use $\mathcal{T}(\bm{x}_0\mathtt{q}) = \mathcal{T}(\bm{x}_0)\mathcal{T}(\mathtt{q})$. Thus, $[q_0,q_1,q_2,q_3]^\top$ lives in a one-dimensional subspace of $\mathbb{R}^4$. Combining with $\mathtt{q}\in\mathbb{T}_\mathbb{Q}$, there are only two feasible $\mathtt{q}$ in the form of $\{\hat{\mathtt{q}},-\hat{\mathtt{q}}\}$, and $\pm \bm{x}_0\hat{\mathtt{q}}$ obviously corresponds to $\pm\bm{x}$. \hfill $\square$ 

\begin{rem}
We remark that for color images with red, green, blue channels, each pixel  has non-negative imaginary parts in its pure quaternion representation. Thus, the ambiguity of the sign ($\pm1$) can be further removed, meaning that the image can be exactly reconstructed. 
\label{ambig}
\end{rem}
\vspace{2mm}
 
Note that $\rank\big(\mathcal{V}(\bm{x})\big)=3$ is often very minor in application, we thus impose this assumption and define  \begin{equation}\label{errsign}
     \dist_p(\bm{z},\bm{x}) = \min\{\|\bm{z}+\bm{x}\|,\|\bm{z}-\bm{x}\|\}
\end{equation} to measure the reconstruction error. Note that the     convergence guarantee for QWF is under the error metric $\dist(\bm{z}_t,\bm{x})=\|\bm{z}_t - \bm{x}\phi(\bm{z}_t)\|$, or equivalently, $\{\bm{z}_t\overline{\phi(\bm{z}_t)}\}$ linearly converges to $\bm{x}$, but the issue is that $\overline{\phi(\bm{z}_t)} = \sign(\bm{z^*}_t\bm{x})$ can never be determined (as it involves the unknown signal $\bm{x}$). Thus, additional efforts are needed to design an algorithm for phase retrieval of $\bm{x}\in \mathbb{Q}_p^d$ with   convergence guarantee regarding $\dist_p(\bm{z}_t,\bm{x})$.

Our idea here is to estimate it up to a sign based on the pure quaternion priori. More precisely, for some $\bm{z}$ our strategy is to find a quaternion phase factor $\mathtt{q}\in \mathbb{T}_\mathbb{Q}$ such that $\bm{z} \mathtt{q}$ is closest to   pure quaternion signal, i.e., 
\begin{equation}
\label{estimatephase}
    \hat{\mathtt{q}} = \mathrm{arg}\min_{\mathtt{q}\in \mathbb{T}_\mathbb{Q}} \|\Re(\bm{z}\mathtt{q})\|,
\end{equation}
and then we map $\bm{z}$ to $\Im(\bm{z}\hat{\mathtt{q}})$. Specilized to the current iteration point $\bm{z}_t$, we find $\mathtt{q}_t$ as follows
\begin{equation}
\label{4.2}
    \mathtt{q}_t = \mathrm{arg}\min_{ \mathtt{q}\in\mathbb{T}_\mathbb{Q}} \|\Re(\bm{z}_t\mathtt{q})\|
\end{equation}
and then map $\bm{z}_t$ to $\Im(\bm{z}_t\mathtt{q}_t)$. A simple observation is  $\|\Re(\bm{z}_t\mathtt{q})\|= \|\mathcal{V}(\bm{z}_t)\mathcal{V}^\top(\overline{\mathtt{q}})\|$, thus (\ref{4.2}) is equal to finding the eigenvector with respect to the smallest eigenvalue of $\mathcal{V}(\bm{z}_t)^\top\mathcal{V}(\bm{z}_t)\in \mathbb{R}^{4\times 4}$, which can be   implemented efficiently  without incurring   computational complexity higher than QWF.

The following condition on $\bm{x}$ assumes a scaling slightly stronger than $\rank\big(\mathcal{V}(\bm{x})\big)=3$, i.e.,  the third singular value of $\mathcal{V}(\bm{x})$ is bounded away from $0$. We restrict our algorithmic analysis to the set of pure quaternion signals with Condition \ref{con4} for some absolute constant $\kappa_0$. 

\begin{con}
\label{con4}
The pure quaternion signal $\bm{x}$ has unit $\ell_2$ norm, and $\mathcal{V}(\bm{x})$ has three positive singular values bounded below by (i.e., larger than) some absolute constant $\kappa_0$ ($\kappa_0>0$).
\end{con}

The following Lemma shows $\mathtt{q}_t$ found by (\ref{4.2}) can transfer the error metric from $\dist(\bm{z}_t,\bm{x})$ to $\dist_p(\bm{z}_t\mathtt{q}_t, \bm{x})$, with   the error preserved up to a multiplicative constant only related to $\kappa_0$ in Condition \ref{con4}.

\begin{lem}
\label{lem77}
Under Condition \ref{con4} we assume $\dist(\bm{z}_t,\bm{x})\leq \delta$ for $\delta\in (0,\frac{1}{2})$. If the quaternion phase factor $\mathtt{q}_t$ is found by (\ref{4.2}), it holds that $\dist_p(\bm{z}_t\mathtt{q}_t,\bm{x})\leq (\frac{6}{\kappa_0}+1)\dist(\bm{z}_t,\bm{x})$.
\end{lem}

\noindent{\it Proof.} We let $\bm{\hat{z}}_t = \bm{z}_t\overline{\phi(\bm{z}_t)}$, then $\dist(\bm{z}_t,\bm{x}) = \|\bm{\hat{z}}_t -\bm{x}\| = \|\mathcal{V}(\bm{\hat{z}}_t)-\mathcal{V}(\bm{x})\|_F$. Hence, $\|\bm{\hat{z}}_t\|\leq \| \bm{\hat{z}}_t-\bm{x}\| +\|\bm{x}\|\leq 1+\delta\leq \frac{3}{2}$. Evidently, $\mathtt{w}_t = \phi(\bm{z}_t)\mathtt{q}_t$ is the solution of $\min_{\mathtt{w}\in \mathbb{T}_{\mathbb{Q}}}\|\Re(\bm{\hat{z}}_t\mathtt{w})\|$. Note that $|\mathtt{w}_t|=1$, we have \begin{equation}
    \begin{aligned}
    \label{4.3}
     &\dist_p(\bm{z}_t\mathtt{q}_t,\bm{x}) = \dist_p(\bm{\hat{z}}_t\mathtt{w}_t,\bm{x}) \\&= \min\{\|\bm{\hat{z}}_t\mathtt{w}_t-\bm{x}\|,\|\bm{\hat{z}}_t\mathtt{w}_t+\bm{x}\|\} \\
     &\leq \min\{\|\bm{\hat{z}}_t\mathtt{w}_t -\bm{\hat{z}}_t \|,\|\bm{\hat{z}}_t\mathtt{w}_t +\bm{\hat{z}}_t\|\}+ \|\bm{\hat{z}}_t-\bm{x}\|\\&\leq\frac{3}{2}\min\{|\mathtt{w}_t-1|,|\mathtt{w}_t+1|\} + \dist(\bm{z}_t,\bm{x})\\
     &\leq \frac{3}{2}|\Im(\mathtt{w}_t)| + \frac{3}{2} \min\{1-\Re(\mathtt{w}_t),1+\Re(\mathtt{w}_t)\}\\&~~~~~~~~~~~~~~~~~~~~~~~~~+\dist(\bm{z}_t,\bm{x})\\
     &\leq \frac{3}{2}|\Im(\mathtt{w}_t)| +  \frac{3}{2}|\Im(\mathtt{w}_t)|^2 + \dist(\bm{z}_t,\bm{x}) \\&\leq 3|\Im(\mathtt{w}_t)|+\dist(\bm{z}_t,\bm{x}).
    \end{aligned}
\end{equation}
Moreover, we have 
\begin{equation}
    \begin{aligned}
    \nonumber
     &\|\Re(\bm{x}\mathtt{w}_t)\| \leq \|\Re((\bm{x}-\bm{\hat{z}}_t)\mathtt{w}_t)\|+\|\Re  (\bm{\hat{z}}_t \mathtt{w}_t)\| \leq \dist(\bm{z}_t,\bm{x})\\&+ \|\Re(\bm{\hat{z}}_t)\| \leq \dist(\bm{z}_t,\bm{x}) +\|\Re(\bm{\hat{z}}_t-\bm{x})\|\leq 2 \dist(\bm{z}_t,\bm{x}),
    \end{aligned}
\end{equation}  where we use the optimality of $\mathtt{w}_t$ in the second inequality. On the other hand, by Condition \ref{con4} $\|\Re(\bm{x}\mathtt{w}_t)\| = \|\mathcal{V}(\bm{x})\mathcal{V}_1(\mathtt{w}_t)\|\geq \kappa_0|\Im(\mathtt{w}_t)|$. Combining these two relations, we obtain $|\Im(\mathtt{w}_t)|\leq \frac{2}{\kappa_0}\dist(\bm{z}_t,\bm{x})$. Substitute this into  (\ref{4.3}) completes the proof. \hfill $\square$

\vspace{2mm}

Now we are at a position to propose the PQWF algorithm. The core spirit is to      pick some positive integer $T_p$ and then invoke the pure quaternion prior every $T_p$ QWF iterations.


\begin{algorithm}
    \caption{Pure Quaternion Wirtinger Flow (PQWF)}\label{alg3}
    \begin{algorithmic}[1]
        \Statex \textbf{Input:} data $(\bm{\alpha}_k,y_k)_{k=1}^n$, step size $\eta$, parameter $T_p$, iteration number $T$
       
        \item   We compute $\bm{z}_0$ as in  Algorithm \ref{alg1}.
        \item   \textbf{for} $i=0,1,...,T-1$:

        \textbf{for} $j=0,1,...,T_p-1$:

         compute $\nabla f(\bm{z}_{iT_p+j})$ as in (\ref{nablala}), then update \begin{equation}
         \begin{aligned}
            &\bm{z}_{iT_p+j+1}= \bm{z}_{iT_p+j} -\eta \nabla f(\bm{z}_{iT_p+j}).
            \end{aligned}
        \end{equation}

        \textbf{end for}

        compute $\mathtt{q}_{(i+1)T_p}$ as the solution of (\ref{estimatephase}) with $\bm{z}=\bm{z}_{(i+1)T_p}$, then we replace $\bm{z}_{(i+1)T_p}$ as follows: \begin{equation}
            \label{minimiza}\bm{z}_{(i+1)T_p}\leftarrow \bm{\tilde{z}}_{(i+1)T_p}:=\Im(\bm{z}_{(i+1)T_p}\cdot \mathtt{q}_{(i+1)T_p}).
        \end{equation}
        \textbf{end for}
         \Statex \textbf{Output:} $\bm{z}_{TT_p}$
    \end{algorithmic}
\end{algorithm}

The next Theorem presents similar linear convergence for   PQWF.

\begin{theorem}
\label{theorem3}
We consider a fixed signal $\bm{x}$ satisfying Condition \ref{con4}. Suppose $\bm{x}$ satisfies Condition \ref{con4}, and by using  $\bm{A}\sim\mathcal{N}_{\mathbb{Q}}^{n\times d}$ and $y_k=|\bm{\alpha}_k^*\bm{x}|^2$ we run Algorithm \ref{alg3} with step size $\eta= O(\frac{1}{d})$. If 
 $n\geq C_0d\log n$ for some $C_0$, $T_p\geq \frac{-2\log(c_2)}{\log(1-c_1/d)}$,   then with high probability  as  in Theorem \ref{theorem2}, for any $k\geq 0$ we have 
\begin{equation}
\label{4.4}
    \dist^2_p(\bm{\tilde{z}}_{(k+1)T_p},\bm{x})\leq \Big(1-\frac{c_1}{4d}\Big)^{T_p}\dist^2_p(\bm{\tilde{z}}_{kT_p},\bm{x}).
\end{equation}
\end{theorem}

\noindent{\it Proof.}  
By assumption we can assume (\ref{3.199}) in Theorem \ref{theorem2} holds for some $c_1$. 
Based on this, we start from $\bm{\tilde{z}}_{kT_p}$, $T_p$ QWF updates give $\bm{z}_{(k+1)T_p}$, since (\ref{3.199}) holds for some $c_1$, we have $$\dist^2(\bm{z}_{(k+1)T_p},\bm{x})\leq \big(1-\frac{c_1}{d}\big)^{T_p}\dist^2(\bm{\tilde{z}}_{kT_p},\bm{x}).$$ Moreover, further using Lemma \ref{lem77}, there exists some  $c_2> 1$ such that \begin{equation}
    \begin{aligned}\nonumber
        \dist_p^2(\bm{\tilde{z}}_{(k+1)
T_p},\bm{x})&\leq \dist_p^2(\bm{z}_{(k+1)T_p}\cdot \mathtt{q}_{(k+1)T_p},\bm{x})\\&\leq c_2\dist^2(\bm{z}_{(k+1)T_p},\bm{x}).
    \end{aligned}
\end{equation} Combining them, we obtain \begin{equation}
    \begin{aligned}
    \dist^2_p(\bm{\tilde{z}}_{(k+1)T_p},\bm{x}) &\leq c_2\big(1-\frac{c_1}{d}\big)^{T_p}\dist^2(\bm{\tilde{z}}_{kT_p},\bm{x})\\& \leq  \big(1-\frac{c_1}{d}\big)^{T_p/2}\dist^2(\bm{\tilde{z}}_{kT_p},\bm{x}),
    \end{aligned}
\end{equation}
where we use $T_p\geq \frac{-2\log(c_2)}{\log(1-c_1/d)}$ in the last inequality. Further use $\sqrt{1-\frac{c_1}{d}}\leq 1-\frac{c_1}{4d}$ and $\dist(\bm{a},\bm{b})\leq \dist_p(\bm{a},\bm{b})$, the result follows. \hfill $\square$


\section{Variants of Quaternion Wirtinger Flow}
	
  Since the seminal work of Wirtinger flow \cite{candes2015phase}, there appeared some variants that refine WF from different respects, among which representatives include   truncated Wirtinger flow (TWF) \cite{chen2017solving}, truncated
amplitude flow (TAF) \cite{wang2017solving}. For example, by a truncation technique, in TWF both spectral initialization and WF update are conducted in a more selective manner. 
Motivated by these developments, we also propose their quaternion versions that we abbreviate as QTWF, QTAF.  We  will numerically test their efficacy. We do not pursue a theoretical analysis (indeed, even in the original works of \cite{chen2017solving,wang2017solving}, the authors only analysed the algorithms in the real case).

\subsection{Quaternion Truncated Wirtinger Flow (QTWF)}

We first propose QTWF. Following \cite{chen2017solving}, we consider the maximum likelihood estimate under Possion noise: 
$$
\max_{\bm{z}\in \mathbb{Q}^d}\mathcal{L}(\bm{z}) :=  \frac{1}{n}\sum_{k=1}^n y_k\log(|\bm{a^*}_k\bm{z}|^2) - |\bm{a^*}_k\bm{z}|^2 .
$$
By Chain rule and Table IV in \cite{xu2015theory}, we obtain\footnote{Here, we can assume $|\bm{\alpha}_k^*z|>0$ for all $k$ since the gradient would be trimmed by $\mathcal{E}_1$ below. This is also true for the QTAF algorithm below.} 
\begin{equation}
    \label{qtwfgradient}
    \nabla\mathcal{L}(\bm{z}) = \Big(\frac{\partial \mathcal{L}}{\partial \bm{z}}\Big)^* = \frac{1}{2n}\sum_{k=1}^n \Big(\frac{y_k}{|\bm{\alpha}_k^*\bm{z}|^2}-1\Big)\bm{\alpha}_k\bm{\alpha}_k^*\bm{z}.
\end{equation} 
We need some pre-specified selection paramters $\theta_z^{\mathrm{lb}},\theta_z^{\mathrm{ub}},\theta_h,\theta_y$. Note that  the spectral initialization  is constructed similarly to QWF, except that the data matrix is constructed more selectively as  $(\lambda_0 =  ({\sum_{k=1}^n y_k}/{n})^{1/2})$:
\begin{equation}
    \label{selectdata}
    \bm{\widetilde{S}}_{in} = \frac{1}{n}\sum_{k=1}^n y_k \bm{\alpha}_k\bm{\alpha^*}_k \mathbbm{1}_{\{|y_k|\leq \theta_y^2\lambda_0^2\}}.
\end{equation}
The QWF update is also modified to be more selective by truncation. Specifically, to update current iteration point $\bm{z}_t$, we let $K_t = \frac{1}{n}\sum_{k=1}^n|y_k - |\bm{\alpha^*}_k\bm{z}_t|^2|$ and will only use
measurements in $\mathcal{E}_1\cap\mathcal{E}_2$ to construct the gradient, where 
\begin{equation}
    \begin{aligned}\label{tuneset}
        &~~~~~\mathcal{E}_1(\bm{z}) := \Big\{k:\theta_z^{\mathrm{lb}} \leq  \frac{|\bm{\alpha^*}_k\bm{z}|}{\|\bm{z}\|}\leq \theta_z^{\mathrm{ub}}\Big\},\\
        &\mathcal{E}_2(\bm{z}) := \Big\{k:\big|y_k - |\bm{\alpha^*}_k\bm{z}|^2\big|\leq \theta_h K_t \frac{|\bm{\alpha^*}_k\bm{z}|}{\|\bm{z}\|}\Big\}.
    \end{aligned}
\end{equation}
Compared with (\ref{qtwfgradient}), we define the trimmed gradient as 
 \begin{equation}
     \label{trimmedtwf}
     \nabla_t\mathcal{L}(\bm{z})=\frac{1}{2n}\sum_{k\in \mathcal{E}_1(\bm{z})\cap\mathcal{E}_2(\bm{z}) } \Big(\frac{y_k}{|\bm{\alpha}_k^*\bm{z}|^2}-1\Big)\bm{\alpha}_k\bm{\alpha}_k^*\bm{z}.
 \end{equation}

\begin{algorithm}
    \caption{Quaternion Truncated Wirtinger Flow (QTWF)}\label{alg4}
    \begin{algorithmic}[1]
        \Statex \textbf{Input:} data $(\bm{\alpha}_k,y_k)_{k=1}^n$, step size $\eta$, iteration number $T$, selection parameters $(\theta_z^{lb},\theta_z^{ub},\theta_h,\theta_y)$ 

        \item  Let $\lambda_0= (\sum_{k=1}^n y_k/n)^{1/2}$. Compute the normalized eigenvector corresponding to the largest standard eigenvalue of (\ref{selectdata}) and denote it by $\bm{\tilde{\nu}}_{in}$ ($\|\bm{\tilde{\nu}}_{in}\|=1$). We use $\bm{z}_0=\lambda_0\cdot \bm{\tilde{\nu}}_{in}$ as initialization.

        \item  \textbf{for} $i=0,1,...,T-1$:

        We compute $\nabla_t\ell(\bm{z}_i)$ as in (\ref{trimmedtwf}) and  update $\bm{z}_i$ to $\bm{z}_{i+1}=\bm{z}_{i}+\eta\nabla_t\mathcal{L}(\bm{z}_i)$.

       \noindent \textbf{end for}

        \Statex \textbf{Output:} $\bm{z}_T$
    \end{algorithmic}
\end{algorithm}

\subsection{Quaternion Truncated Amplitude Flow (QTAF)}
As in \cite{wang2017solving}, QTAF is based on the amplitude-based model $y'_k = |\bm{\alpha}_k^*\bm{x}|$ (hence $y_k'=\sqrt{y_k}$), and the goal is to minimize the corresponding $\ell_2$ loss
$$
    \min_{\bm{z}\in \mathbb{Q}^d}\ell(\bm{z}): = \frac{1}{n}\sum_{k=1}^n \big(|\bm{\alpha}_k^*\bm{z}| - y'_k\big)^2.
$$
By chain rule and Table IV in \cite{xu2015theory} we obtain \begin{equation}
    \nabla\ell(\bm{z})= \Big(\frac{\partial \ell}{\partial \bm{z}}\Big)^*=\frac{1}{2n}\sum_{k=1}^n \bm{\alpha}_k \big(\bm{\alpha}_k^*\bm{z} - y_k'\frac{\bm{\alpha}_k^*\bm{z}}{|\bm{\alpha}_k^*\bm{z}|}\big). 
\end{equation}
QTAF involves two tuning parameters $\gamma,\rho\in (0,1)$. It adopts the totally different 
orthogonality-promoting initialization. Specifically, we define $\overline{\mathcal{I}}_0\subset [n]$ as the indices corresponding to the $\lceil \rho n\rceil$ largest values of $y_k'/ \|\bm{\alpha}_k\|$, and for initialization it uses the data matrix as \begin{equation}
    \widehat{\bm{S}}_{in} = \frac{1}{|\overline{\mathcal{I}}_0|}\sum_{k\in \overline{\mathcal{I}}_0}\frac{\bm{\alpha}_k\bm{\alpha}_k^*}{\|\bm{\alpha}_k\|^2}.\label{afdatamat}
\end{equation} 
The parameter ${\gamma}$ is used to trim the gradient. Specifically, we define $\mathcal{I}_{\bm{z}}:=\{k\in [n]:|\bm{\alpha}_k^*\bm{z}|\geq y_k'/(1+\gamma)\}$ and further the trimmed gradient \begin{equation}\label{tunegraaf}
    \nabla_t\ell(\bm{z})=\frac{1}{2n}\sum_{k\in \mathcal{I}_{\bm{z}}}\Big(1-\frac{y_k'}{|\bm{\alpha}_k^*\bm{z}|}\Big) \bm{\alpha}_k\bm{\alpha}_k^*\bm{z}.
\end{equation}
\begin{algorithm}
    \caption{Quaternion Truncated Amplitude Flow (QTAF)}\label{alg5}
    \begin{algorithmic}[1]
        \Statex \textbf{Input:} data $(\bm{\alpha}_k,y_k':=|\bm{\alpha}_k^*\bm{x}|)_{k=1}^n$, step size $\eta$, iteration number $T$, selection parameters $(\gamma,\rho)$.

        \item Let $\lambda_0=(\sum_{k=1}^n (y_k')^2/n)^{1/2}$. Compute the  normalized eigenvector corresponding to the largest standard eigenvalue of (\ref{afdatamat}) and denote it by $\widehat{\bm{\nu}}_{in}$ ($\|\widehat{\bm{\nu}}_{in}\|=1$). Then we obtain $\bm{z}_0=\lambda_0\cdot \widehat{\bm{\nu}}_{in}$ as initialization.

        \item \textbf{for} $i=0,1,...,T-1:$

        We compute $\nabla_t\ell(\bm{z}_i)$ as in (\ref{tunegraaf}) and then update $\bm{z}_i$ to $\bm{z}_{i+1}=\bm{z}_i-\eta\nabla_t\ell(\bm{z}_i)$.

        \noindent \textbf{end for}

        \Statex \textbf{Output:} $\bm{z}_T$
    \end{algorithmic}
\end{algorithm} 
\subsection{The Pure Quaternion Versions}
The developed techniques for utilizing a pure quaternion priori  
can be similarly incorporated into QTWF, QTAF --- by mapping    $\bm{z}_t$ to $\Im(\bm{z}_t\mathtt{q}_t)$ ($\mathtt{q}_t$ is defined in (\ref{4.2}))  every $T_p$ iterations. For clarity, we present Pure QTWF (PQTWF) and Pure QTAF (PQTAF) in the following. 
\begin{algorithm}
    \caption{\textcolor{black}{Pure Quaternion Truncated Wirtinger Flow (PQTWF)}}\label{alg6}
    \begin{algorithmic}[1]
        \Statex \textbf{Input:} data $(\bm{\alpha}_k,y_k)_{k=1}^n$, step size $\eta$, parameter $T_p$, iteration $T$, selection parameters $(\theta_z^{lb},\theta_z^{ub},\theta_h,\theta_y)$

         \item  The initialization   is the same as step 1 in Algorithm \ref{alg4}.

        \item  \textbf{for} $i=0,1,...,T-1$:

        \textbf{for} $j=0,1,...,T_p-1:$

          Compute $\nabla_t\ell(\bm{z}_{iT_p+j})$ as in (\ref{trimmedtwf}) and  update $\bm{z}_i$ to $\bm{z}_{iT_p+j+1}=\bm{z}_{{iT_p+j}}+\eta\nabla_t\mathcal{L}(\bm{z}_{iT_p+j})$.

\textbf{end for}

  compute $\mathtt{q}_{(i+1)T_p}$ as the solution of (\ref{estimatephase}) with $\bm{z}=\bm{z}_{(i+1)T_p}$, then we replace $\bm{z}_{(i+1)T_p}$ as follows: \begin{equation}
            \bm{z}_{(i+1)T_p}\leftarrow  \bm{\tilde{z}}_{(i+1)T_p}:=\Im(\bm{z}_{(i+1)T_p}\cdot \mathtt{q}_{(i+1)T_p}).
        \end{equation}

       \noindent \textbf{end for}

        \Statex \textbf{Output:} $\bm{z}_{TT_p}$
    \end{algorithmic}
\end{algorithm}

\begin{algorithm}
        \caption{\textcolor{black}{Pure Quaternion Truncated Amplitude   Flow (PQTAF)}}\label{alg7}
    \begin{algorithmic}[1]
\Statex \textbf{Input:} data $(\bm{\alpha}_k,y_k':=|\bm{\alpha}_k^*\bm{x}|)_{k=1}^n$, step size $\eta$, parameter $T_p$, iteration $T$, selection parameters $(\gamma,\rho)$

 \item  The initialization   is the same as step 1 in Algorithm \ref{alg5}.

 \item \textbf{for} $i=0,1,...,T-1$:

        \textbf{for} $j=0,1,...,T_p-1:$

          Compute $\nabla_t\ell(\bm{z}_{iT_p+j})$ as in (\ref{tunegraaf}) and  update $\bm{z}_i$ to $\bm{z}_{iT_p+j+1}=\bm{z}_{{iT_p+j}}+\eta\nabla_t\mathcal{L}(\bm{z}_{iT_p+j})$.

\textbf{end for}

 compute $\mathtt{q}_{(i+1)T_p}$ as the solution of (\ref{estimatephase}) with $\bm{z}=\bm{z}_{(i+1)T_p}$, then we replace $\bm{z}_{(i+1)T_p}$ as follows: \begin{equation}
            \label{minimi}\bm{z}_{(i+1)T_p}\leftarrow \bm{\tilde{z}}_{(i+1)T_p}:=\Im(\bm{z}_{(i+1)T_p}\cdot \mathtt{q}_{(i+1)T_p}).
        \end{equation}

       \noindent \textbf{end for}

         \Statex \textbf{Output:} $\bm{z}_{TT_p}$
    \end{algorithmic}
\end{algorithm}
\section{Experimental Results}\label{sec6}
We present experimental results in this Section, specifically Sections \ref{E5.1}, \ref{variant}, \ref{fewer} for synthetic data, and Section \ref{E5.2} for color images.\footnote{Our implementation is based on the quaternion toolbox for Matlab developed by S.~J.~Sangwine and N.~Le~Bihan available in \url{https://sourceforge.net/projects/qtfm/}.}

\subsection{Synthetic Data}\label{E5.1}
In each single trial of QWF, we use Gaussian measurement ensemble $\bm{A}\sim \mathcal{N}_\mathbb{Q}^{n\times d}$. The entries of quaternion signal $\bm{x}\in \mathbb{Q}^{d}$ (resp. pure quaternion signal $\bm{x}\in \mathbb{Q}_p^{d}$) are i.i.d. copies of $\mathcal{N}(0,1)+\sum_{\vartheta=\ii,\jj,\kk} \mathcal{N}(0,1)\vartheta$ (resp. $\sum_{\vartheta=\ii,\jj,\kk} \mathcal{N}(0,1)\vartheta$),  then   $\bm{x}$ will be normalized so that $\|\bm{x}\| = 1$.  We apply $100$ power iterations   to approximately find the $\bm{\nu}_{in}$ as the leading eigenvector of the Hermitian $\bm{S}_{in}$. Here, the power method for quaternion Hermitian matrix is parallel to the complex case \cite{li2019power}. By default, we set $\eta = \frac{0.2n}{\sum_{k=1}^ny_k}$  and run $1500$ QWF updates to obtain the reconstructed signal $\bm{\hat{z}}$.

We first report the success rate of QWF under different sample sizes (Figure \ref{fig1}(a)). Specifically, we test $d = 100$ and the sample sizes $\frac{n}{d} = 3:0.5:13$. For each $n$ we conduct $100$ independent trials, with a  trial  claimed to be  success if $\dist(\bm{\hat{z}},\bm{x})<10^{-5}$ \cite{candes2015phase,chen2017solving}. One can see that the phase transition starts from $n/d=6.5$, and the success rate reaches $1$ at $n/d=9$. Under $n\geq 9d$ the success rate remains  $1$. In Figure \ref{fig1}(b), we also plot the curve of $\log(\dist(\bm{z}_t,\bm{x}))$ versus $t$ in a single trial with $m/n=10$ and $3500$ QWF updates. The curve decreases and shapes like a straight line, and then reaches a plateau, which is consistent with our linear convergence guarantee in Theorem \ref{theorem2}.  

\begin{figure}[!t]
    \centering
    \includegraphics[scale = 0.43]{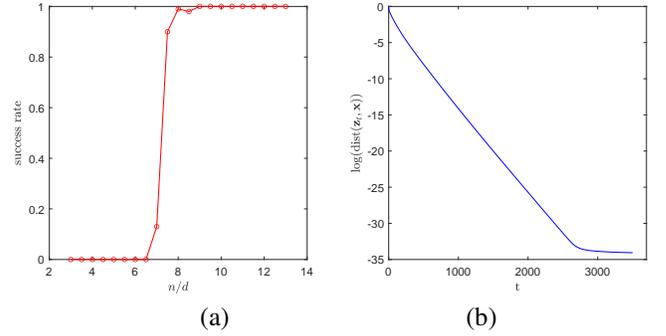}

 ~~~   (a) \hspace{2.9cm} (b) 

    \caption{(a): success rate of QWF; (b): linear convergence.}
    \label{fig1}
\end{figure}

Then we go into pure quaternion signal reconstruction via PQWF, and we test $\bm{x}\in \mathbb{Q}_p^{50}$. Recall that   the error metric   becomes $\dist_p(\bm{z} ,\bm{x})= \min\{\|\bm{z}+\bm{x}\|,\|\bm{z}-\bm{x}\|\}$  since the synthetic signal   admits Condition \ref{con4}.  Our main goal is to  show the proposed PQWF in Section \ref{sec4}  can effectively incorporate the pure quaternion priori and  gain notable benefits from it, e.g., earlier phase transition. We also try to reveal the significant role played by the phase factor estimate (\ref{4.2}). \textcolor{black}{For this purpose, we invite the following algorithms to compete with PQWF:}
\begin{itemize}
    \item \textbf{Algorithm I} (Alg. I): We run QWF and finally turn its output to pure quaternion via phase factor estimate. That is, it first runs Algorithm \ref{alg2} to get $\bm{{z}}_T$, and then takes $\Im(\bm{{z}}_T\mathtt{q}_T)$ as solution, where $\mathtt{q}_T$ is the solution to (\ref{4.2}) with $\bm{z}_t=\bm{{z}}_T$. 
    \item \textbf{Algorithm II} (Alg. II): This algorithm is a variant of PQWF --- it incorporates the pure quaternion prior every $T_p$ iterations by directly removing the real part without phase factor estimate. That is, it runs Algorithm \ref{alg3} but with   (\ref{minimiza}) substituted by $\bm{z}_{(i+1)T_p}=\Im(\bm{z}_{(i+1)T_p})$. 

    \item \textbf{Algorithm III} (Alg. III): We run QWF and finally turn its output to pure quaternion by directly removing the real part without phase factor estimate. That is, it first runs Algorithm \ref{alg2} to get $\bm{z}_T$ and takes $\Im(\bm{z}_T)$ as solution. 
\end{itemize}
The experimental results are reported in Figure \ref{fig2}. 

\begin{figure*}[!t]
    \centering
    \includegraphics[scale = 0.5]{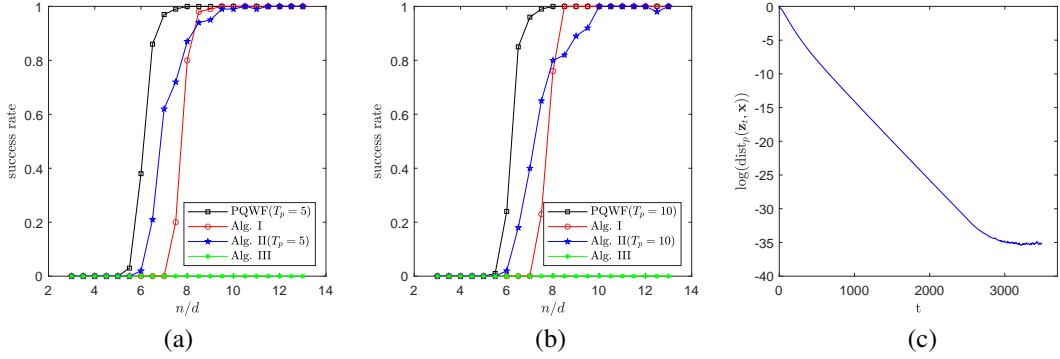}

 ~~~   (a) \hspace{4.3cm} (b)  \hspace{4.3cm} (c) 

    \caption{(a) and (b): Success rate of PQWF and Alg. I--III  with $T_p=5$ or $T_p = 10$ (if applicable); (c): linear convergence.}
    \label{fig2}
\end{figure*}

We provide the success rate of PQWF and Alg. I--III under $\frac{m}{n}=3:0.5:13$ in Figure \ref{fig2}(a) (for $T_p = 5$) and Figure \ref{fig2}(b) (for $T_p = 10$). 
Clearly, the proposed PQWF outperforms   Alg. I--III. For instance, it embraces a phase transition earlier than Alg I, thus confirming the advantage of PQWF in utilizing the pure quaternion priori to reduce the measurement number. In stark contrast,  imposing the pure quaternion constraint by   removing the real part, Alg. II and Alg. III have worse performances, which demonstrates the crucial role played by the phase factor estimate (\ref{4.2}). Since $\bm{\hat{z}}$ returned by QWF   only recovers $\bm{x}$ up to the unknown right quaternion phase factor, $\Im(\bm{\hat{z}})$  obviously does not approximate $\bm{x}$ up to a sign. This explains why the success rate of Alg. III  remains zero. On the other hand, it would be more interesting to take a closer look at the   curve of Alg. II. In particular, Alg. II also enjoys an   phase transition earlier than Alg. I ($6\leq\frac{n}{d}\leq 7$). 
However, under relatively sufficient measurements (e.g., $8.5\leq\frac{n}{d}\leq 10$, Figure \ref{fig2}(b)), it can be even worse than Alg. I. For illustration, we comment that the QWF update is based on the data, while removing the real part is based on the priori, so their effects on the iteration are likely to somehow neutralize, which possibly explains the unsatisfactory performance of Alg. II.  Therefore, removing the real part is not a sensible starategy for utilizing the pure quaterion priori. Beyond that, we also track $\log(\dist_p(\bm{\tilde{z} }_{5k},\bm{x}))$ in a single trial of PQWF (with $T_p = 5$, $m/n=8$) and   plot the error decreasing curve in Figure \ref{fig2}(c). This   corroborates our   linear convergence guarantee.



\subsection{Variants of Quaternion Wirtinger Flow}\label{variant}

Recall that we have proposed the more refined algorithms QTWF and QTAF.
 To see their efficacy, we compare the   success rate of QWF, QTWF, QTAF under $\frac{n}{d} =3:0.5:13$, and the underlying $\bm{x}\in \mathbb{Q}^{50}$ is randomly drawn as previous experiments. Each success rate is based on 100 independent trials.  We use $\eta = 0.2$ for QWF,  $(\theta^{lb}_z,\theta^{ub}_z,\theta_h,\theta_y,\eta ) = (0.3,4.5,5,3 ,0.8)$ for QTWF, $(\gamma,\rho ,\eta) = (0.8,\frac{1}{6},1.2)$ for QTAF. The results are shown in Figure \ref{fig3}. Evidently, at the cost of more parameters and more complicated algorithms, QTWF and QTAF perform notably better than QWF. 

\begin{figure}[!t]
    \centering
    \includegraphics[scale = 0.42]{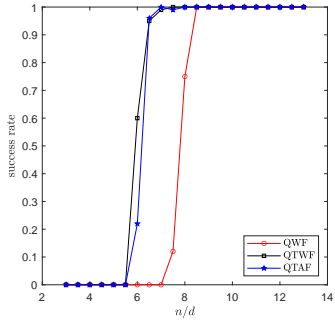}
    \caption{ Success rate of QWF, QTWF, QTAF.}
    \label{fig3}
\end{figure}
 
\begin{figure*}[!t]
    \centering
    \includegraphics[scale = 0.46]{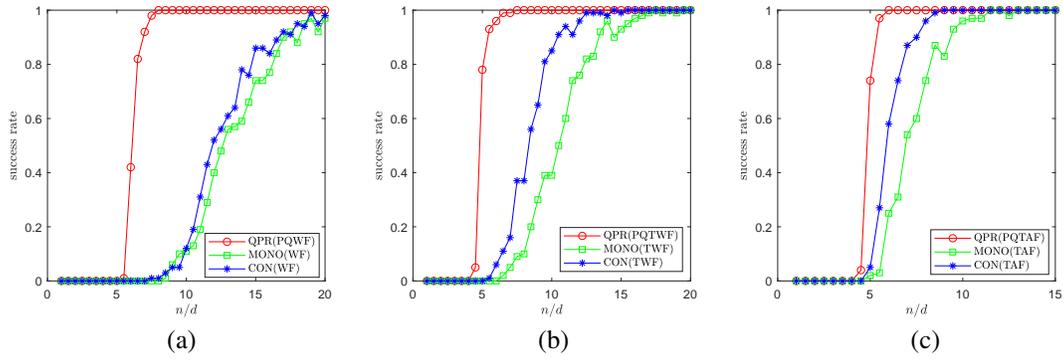}

 ~~~   (a) \hspace{4.3cm} (b)  \hspace{4.3cm} (c) 

    \caption{"QPR": the model in this work; "MONO": the monochromatic model; "CON": the concatenation model. (a): PQWF (Algorithm \ref{alg3}) for "QPR", WF \cite{candes2015phase} for "MONO" and "CON"; (b): PQTWF (Algorithm \ref{alg6}) for "QPR", TWF \cite{chen2017solving} for "MONO" and "CON"; (c): PQTAF (Algorithm \ref{alg7}) for "QPR", TAF \cite{wang2017solving} for "MONO" and "CON".}
    \label{revisionadvan}
\end{figure*}
\begin{figure*}[!t]
    \centering
    \includegraphics[scale = 0.5]{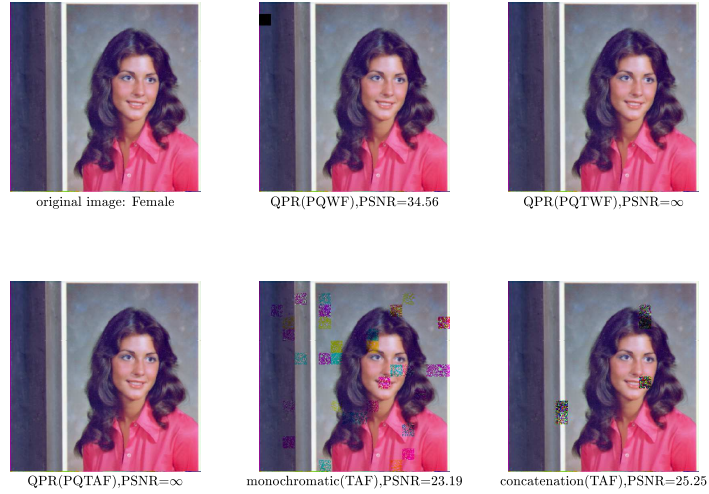}
    \caption{All simulations use $7.5\times 256=1920$ phaseless measurements for each block.  Original image "Female", reconstructed images of QPR using PQWF, PQTWF, PQTAF, and   the monochromatic model, concatenation model using TAF. Note that PSNR$=\infty$ appears because the signals returned by the algorithms have relative error (i.e., $\|\bm{\hat{I}}-\bm{I}\|_F/\|\bm{I}\|_F$ with $\bm{\hat{I}},\bm{I}$ being the reconstructed image, original image (resp.)) less than $10^{-9}$, then after changing to the "uint8" format the reconstructed image exactly equals to the original image.}
    \label{fig4}
\end{figure*}
\subsection{\textcolor{black}{Reduced Measurement Number in   Pure Quaternion Signal Recovery}}\label{fewer}
In this part, we compare QPR with the real methods of monochromatic model and concatenation model in the regime of pure quaternion signal recovery. It will be shown that our quaternion model can succeed using notably fewer phaseless measurements. 
\begin{figure*}[!t]
    \centering
    \includegraphics[scale = 0.5]{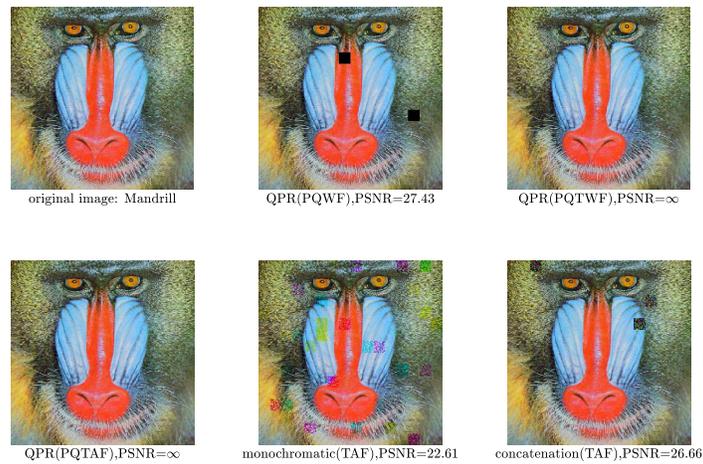}
    \caption{This is an experiment parallel to Figure \ref{fig4}. All simulations use $1920$ phaseless measurements for each block with size $16\times 16$. QPR using PQTWF or PQTAF achieves exact reconstruction, QPR using PQWF fails in two blocks, the two real models using TAF fail in more blocks.}
    \label{fig7}
\end{figure*}

Recall that for $\bm{x}\in \mathbb{Q}_p^{50}$ satisfying Condition \ref{con4}, in QPR we propose PQWF (Algorithm \ref{alg3}) and the more refined PQTWF (Algorithm \ref{alg6}), PQTAF (Algorithm \ref{alg7}) for recovering $\bm{x}$ up to a sign. To achieve this, alternatively one can use the phase retrieval of real signal based on monochromatic model or concatenation model. In this part, we numerically compare QPR with these two real models: 
\begin{itemize}
    \item The \textbf{monochromatic} model conducts phase retrieval for $\mathcal{P}^{\ii}(\bm{x}),\mathcal{P}^{\jj}(\bm{x}),\mathcal{P}^{\kk}(\bm{x})\in \mathbb{R}^{50}$ separately. Let $\bm{\hat{x}}_{\ii},\bm{\hat{x}}_{\jj},\bm{\hat{x}}_{\kk}$ be the corresponding reconstructed signals, the error metric is $\big(\dist_p(\mathcal{P}^{\ii}(\bm{x}),\bm{\hat{x}}_{\ii})^2+\dist_p(\mathcal{P}^{\jj}(\bm{x}),\bm{\hat{x}}_{\jj})^2+\dist_p(\mathcal{P}^{\kk}(\bm{x}),\bm{\hat{x}}_{\kk})^2\big)^{1/2}$. 
    \item The \textbf{concatenation} model conducts phase retrieval for $\bm{x}_{con}:=[(\mathcal{P}^{\ii}(\bm{x}))^\top,(\mathcal{P}^{\jj}(\bm{x}))^\top,(\mathcal{P}^{\kk}(\bm{x}))^\top]^\top\in \mathbb{R}^{150}$. Let $\bm{\hat{x}}$ be the reconstructed signal, the error metric is $\dist_p({\bm{x}_{con}},\bm{\hat{x}})$. 
\end{itemize}
For the above two real models, we use measurement matrix with i.i.d. standard Gaussian entries and test   WF, TWF   and TAF using parameters recommended by the original papers \cite{candes2015phase,chen2017solving,wang2017solving}. Accordingly, we test QPR under PQWF  (step size $\eta=0.15$, other parameters remain the same),  PQTWF (same parameters as QTWF in Section \ref{variant}, $T_p=5$), PQTAF (same parameters as QTAF in Section \ref{variant}, $T_p=5$). 
The pure quaternion signal is drawn as before.

We test the success rate of QPR under our PQWF, PQTWF, PQTAF (Algorithms \ref{alg3}, \ref{alg6}, \ref{alg7}), and then compare with the above two real models under WF \cite{candes2015phase}, TWF \cite{chen2017solving}, TAF \cite{wang2017solving}. The results are displayed in Figure \ref{revisionadvan}(a)-(c) (to be fair, each subfigure compares the considered models under comparable algorithms, see the caption of Figure \ref{revisionadvan} for details).  Observe that QPR enjoys  a full success rate   when $m\geq 8d$ under PQWF,  when $m\geq 7.5 d$ under PQTWF, or when $m\geq 6d$ under PQTAF. By contrast, the real models  require at least $m=9d$ to achieve full success rate, i.e., the concatenation model under TAF, whereas all other cases require much larger sample size.    
These results unveil an advantage of our QPR  model over the two real models --- while each measurement (in the three models) is commonly a positive scalar, when comparable algorithms are used, QPR can succeed with notably fewer measurements.

\subsection{\textcolor{black}{Color Image}}\label{E5.2}
 We test our algorithms in the real data of color image. The tested $256 \times 256$ images "Female"  and "Mandrill" are available online.\footnote{\url{https://sipi.usc.edu/database/database.php?volume=misc\&image=4\#top}} 
To implement phase retrieval with Gaussian measurement matrix, the image is divided into  $256$ blocks of size $16\times 16$.\footnote{To handle the whole image without dividing into blocks (as in \cite{candes2015phase,chen2017solving,wang2017solving}), we need to develop more practical measurement matrix for QPR, e.g., coded diffraction pattern. See more discussions in Section \ref{sec7}.} In our QPR model, each block is modeled as pure quaternion signal in $\mathbb{Q}_p^{256}$, and we will test algorithms PQWF, PQTWF and PQTAF with parameters as before. For comparison, we  also test  the real methods of monochromatic model and concatenation model using TAF (from Figure \ref{revisionadvan}, TAF performs better than TWF and WF for the real models). 
Recall that a color image can be exactly reconstructed without any ambiguity (Remark \ref{ambig}). Using the same measurement number of $m=7.5\times 256$ to deal with each block separately, we show the original image and reconstructed images (with PSNR) using different  models/algorithms in Figures \ref{fig4}, \ref{fig7}. Clearly, using our quaternion model and an oversampling rate of $7.5$, PQWF only fails in one block, while the more refined algorithms PQTWF and PQTAF exactly recover the whole image. By contrast, the two real methods fail in much more blocks and deliver much lower PSNR.   These results agree with the conclusion of Section \ref{fewer}.

\textcolor{black}{Additionally, we further test QPR (using PQTAF) and the real methods of monochromatic model and concatenation model (both using TAF) over 24 color images from  the Kodak24 image dataset.\footnote{\url{https://www.kaggle.com/datasets/sherylmehta/kodak-dataset}} Specifically, we downsample the images to "$192\times 128$"/"$128\times 192$" and  divide them into $96$ blocks of size $16\times 16$, then we perform (quaternion) phase retrieval in each block  using $m=7.5\times 256$ measurements independently. The results are reported as follows: (1) The monochromatic method does not exactly recover any image, with the mean and standard deviation of the 24 PSNR values being $23.54$ and $2.10$, respectively; (2) The concatenation model achieves exact reconstruction over 12 images, with the mean and standard deviation of the remaining 12 PSNR values being $27.20$ and $3.45$, respectively; (3) Our quaternion method delivers exact recovery over 22 images, while the PSNR values of the remaining two images are $26.32$ and $13.34$. Note that QPR exactly recovers most images, which again validates our advantage that QPR succeeds with fewer measurements. However, QPR does not perform well in two images. Taking a closer look at these two images, we find that the blocks where QPR fails do not satisfy Condition \ref{con4},\footnote{See supplementary material for the two images in which QPR fails.} in the sense   that their three channels are (nearly or exactly) real linearly dependent. Note  that Lemma \ref{lem66} no longer stands for pure quaternion signal whose three imaginary parts are (real) linearly dependent,  and in general we cannot identify such signal up to a sign, so QPR unavoidably fails in these blocks due to identifiability issue. Thus, extra cautiousness is needed to ensure Condition \ref{con4} when one applies QPR to pure quaternion signal recovery. Furthermore, we conduct the experiment again with the measurement number for each block reduced to $m=6.5\times 256$ (i.e., under the oversampling rate   $6.5$), then the performances of the two real methods deteriorate significantly, while QPR using PQTAF still exactly recovers 22 color images and fails in the remaining two due to the signal nature. More details are reported in the supplementary material.}

\section{Concluding Remarks}\label{sec7}


In this paper, we initiate the study of quaternion phase retrieval  (QPR) problem, which is formulated as the reconstruction of $\bm{x}\in \mathbb{Q}^d$ from $|\bm{Ax}|^2$ with known $\bm{A}\in \mathbb{Q}^{n\times d}$. As the theoretical foundation,  we   first confirm that the global right quaternion phase factor is the only trivial ambiguity. 
Then, we propose quaternion Wirtinger flow (QWF) as a  scalable and practical algorithm for solving QPR. The linear convergence of QWF has been proved and presented as a major theoretical result. The   technical ingredients involve the HR calculus and some techniques for handling quaternion matrices. While our proof strategy is adjusted from the complex Wirtinger flow \cite{candes2015phase}, a series of different treatments are employed with some new machinery (Remark \ref{rem4} and other remarks in supplementary material).

Special attention is paid to QPR of pure quaternion signal. With an additional minor assumption, one can reconstruct the signal up to the trivial ambiguity of a sign (Lemma \ref{lem66}), which can be further removed in color image recovery (Remark \ref{ambig}). By using a  crucial phase factor estimate, we develop the PQWF algorithm that can effectively utilize the pure quaternion priori. Note that PQWF enjoys guaranteed linear convergence. Motivated by existing refinements of WF, specifically TWF and TAF, we further propose QTWF, QTAF and their pure quaternion versions PQTWF, PQTAF. Their advantages are numerically demonstrated. We provide experimental results on synthetic data and color images that corroborate our theories. A surprising finding is that, for pure quaternion signal recovery, our quaternion method often succeeds with measurements notably fewer than  real methods based on monochromatic model or concatenation model. This is claimed as the advantage of the QPR model and makes our quaternion method preferable in a situation where one can only obtain a very limited number of phaseless measurements.

We note that  this work  only provides a starting point for   the research of QPR, and there are undoubtedly many questions worth further exploration. We point out several directions to close this paper. Theoretically, although $O(d)$ measurements have been shown to be sufficient for QPR, it would be of theoretical interest to investigate the precise measurement number needed for recovery of all signals in $\mathbb{Q}^d$. This is known as the minimal measurement number in phase retrieval (e.g., \cite{balan2006signal,bandeira2014saving,conca2015algebraic,wang2019generalized}). Besides, compared to   $\bm{A}\sim \mathcal{N}_\mathbb{Q}^{n\times d}$, coded diffraction pattern (CDP) that applies Fourier transform to the masked signal would be more practical for some applications \cite{candes2015phase1}. It is   appealing to develop a similar measurement scheme for quaternion signal, ideally accompanied by a guaranteed algorithm. \textcolor{black}{Moreover, while our simulations unveil an advantage of QPR (i.e., success using fewer measurements), it is interesting to explore more privileges  so that one knows in what regime the quaternion model is preferable.} By making use of the (approximately) low-rank structure of color images, the advantages of using quaternion-based method were demonstrated in image restoration or inpainting \cite{chen2019low,jia2019robust,chen2022color}. On the other hand, recent works have explored how to incorporate the low-rank priori into phase retrieval (referred to as low-rank phase retrieval, LRPR) \cite{vaswani2017low,lee2021phase,nayer2020provable}. Taken collectively, we conjecture that it may be fruitful to explore the benefit of using quaternion method in LRPR.





 
%
 \bibliographystyle{IEEEtran}
\bibliography{libr}
\begin{appendix}

\subsection{Auxiliary Results}\label{apenA}

We provide some expectation results to support our analysis. Their proofs and some technical remarks are provided in the supplementary material.  

\begin{lem}
\label{lem3}
Assume entries of $\bm{\alpha}\in \mathbb{Q}^d$ are i.i.d. drawn from $\mathcal{N}_\mathbb{Q}$, and $\bm{u},\bm{v}\in\mathbb{Q}^d$ satisfy $\|\bm{u}\| = \|\bm{v}\|=1$. Then we have the following:

\vspace{1mm}

\noindent{\rm	 (a)} {\rm(rotational invariance)} For any unitary matrix $\bm{P}\in \mathbb{Q}^{d\times d}$, $\bm{P\alpha}$ and $\bm{\alpha}$ have the same distribution.

\vspace{1mm}
\noindent{\rm (b)} $\mathbbm{E} |\bm{\alpha^* u}|^{2l} = \frac{(l+1)!}{2^l}$ for any positive integer $l$.

\vspace{1mm}
\noindent{\rm (c)} $\mathbbm{E}\big[\Re(\bm{u^*\alpha\alpha^* v})\big]^2 = \frac{1}{4} + \frac{5}{4}\big[\Re(\bm{u^*v})\big]^2 - \frac{1}{4}\big|\Im(\bm{u^*v})\big|^2$.

\vspace{1mm}
\noindent{\rm (d)} $\mathbbm{E}\big[\bm{\alpha\alpha^*u}|\bm{\alpha^*v}|^2\big] = \bm{u}+\frac{1}{2}\bm{vv^*u}$.
\end{lem}

 


\begin{lem}
\label{lem6}
Assume $\bm{\gamma}_1,\cdots,\bm{\gamma}_n$ are independent random vectors in $\mathbb{R}^d$ that have entries $[\gamma_{kj}]$ i.i.d. drawn from $\mathcal{N}(0,1)$. Fix $i,j\in \{1,2\}$ and any sufficiently small $\delta>0$, when $n\geq C_1 \delta^{-2}d\log n$ for sufficiently large $C_1$, with probability at least $1-2n^{-9}-2\exp(-cd)$, we have \begin{equation}\nonumber
    \Big\|\frac{1}{n}\sum_{k=1}^n \gamma_{ki}\gamma_{kj} \bm{\gamma}_k\bm{\gamma}_k^\top - \mathbbm{E}\big[\gamma_{ki}\gamma_{kj} \bm{\gamma}_k\bm{\gamma}_k^\top\big]\Big\|\leq C\delta.
\end{equation}
\end{lem}




 \end{appendix}

\newpage

\section{\textbf{Supplementary Material}}

\section{Missing Proofs}
\subsection{Proof for Theorem 1}
\IEEEPARstart{S}{mall}  ball method due to Mendelson \cite{mendelson2015learning,mendelson2018learning}   is an effective tool to lower bound a non-negative empirical process. The  version provided below can be found in Lemma 1 of \cite{dirksen2016gap} (with $p=2$).

\begin{pro}
\label{lemma2}
  {\rm (Mendelson's Small Ball Method)} Assume that $\mathcal{F}$ is a class of functions 
 from $\mathbb{C}^n$ into $\mathbb{C}$.  Let $\bm{\varphi}$ be   random vector on $\mathbb{C}^n$, and $\bm{\varphi}_1,\cdots,\bm{\varphi}_m$ be independent copies of $\bm{\varphi}$. For $u> 0$ we define $Q_\mathcal{F} (u) = \inf_{f\in \mathcal{F}} \mathbbm{P}(|f(\bm{\varphi})|\geq u)$. Let $\varepsilon_1,\cdots,\varepsilon_n$ be independent Rademacher random variables (i.e., $\mathbbm{P}(\varepsilon_k = 1) = \mathbbm{P}(\varepsilon_k = -1) = 1/2$),  we further define $R_m(\mathcal{F}) = \mathbbm{E}\sup_{f\in\mathcal{F}}|\frac{1}{m}\sum_{k=1}^m \varepsilon_kf(\bm{\varphi}_k)|$. Then, for any $u>0$, $t>0$, with probability at least $1-2\exp(-2t^2)$, we have 
 \begin{equation}
     \nonumber
     \inf_{f\in \mathcal{F}} \frac{1}{m}\sum_{k=1}^m |f(\bm{\varphi}_k)|^2 \geq u^2 \Big(\mathcal{Q}_{\mathcal{F}}(2u) -\frac{4R_m(\mathcal{F})}{u}-\frac{t}{\sqrt{m}}\Big).
 \end{equation}
\end{pro}

 \begin{theorem}
Assume $\bm{A}=[\bm{\alpha}_1,\cdots,\bm{\alpha}_n]^*\sim \mathcal{N}_\mathbb{Q}^{n\times d}$. When $n\geq Cd$ for some absolute constant $C$, with probability at least $1-\exp(-C_1n)$, all signals $\bm{x}$  in $\mathbb{Q}^d$ can be reconstructed from $\{|\bm{\alpha^*}_k\bm{x}|^2:k\in [n]\}$ up to a global right quaternion phase factor.  
\end{theorem}

\noindent {\it Proof.}   We only need to show $|\bm{\alpha^*}_k\bm{x}|^2 = |\bm{\alpha^*}_k\bm{y}|^2$ for all $k\in [n]$ leads to $\bm{xx^*}=\bm{yy^*}$, since this  can imply $\bm{x}=\bm{y}\cdot \mathtt{q}$ for some unit quaternion $\mathtt{q}$ (see Lemma 1 in the paper).   By $\Re(\Tr(\bm{AB})) =\Re(\Tr(\bm{BA}))$ for any $\bm{A}\in \mathbb{Q}^{d_1\times d_2}$, $\bm{B}\in \mathbb{Q}^{d_2\times d_1}$, assuming $\bm{xx^*}\neq \bm{yy^*}$, some algebra gives \begin{equation}\tag{29}
    \begin{aligned} 
    \label{3.1}
    &|\bm{\alpha^*}_k\bm{x}|^2 - |\bm{\alpha^*}_k\bm{y}|^2 = \Re\big[ \Tr\big(\bm{\alpha^*}_k(\bm{xx^*}-\bm{yy^*})\bm{\alpha}_k\big)\big] \\&= \|\bm{xx^*}-\bm{yy^*}\|_F\cdot\Re \left<\frac{\bm{xx^*}-\bm{yy^*}}{\|\bm{xx^*}-\bm{yy^*}\|_F},\bm{\alpha}_k\bm{\alpha^*}_k\right>.
    \end{aligned}
\end{equation}
Recall that $\mathcal{H}_{d,r}^\mathbb{Q}$ is the set of all $d\times d$ quaternion Hermitian matrices with rank not exceeding $r$, and we further let $(\mathcal{H}_{d,r}^\mathbb{Q})^*=\mathcal{H}_{d,r}^\mathbb{Q}\cap \{\bm{X} \in \mathbb{Q}^{d\times d}:\|\bm{X}\|_F =1\}$. In this proof, we use the shorthand $\inf_{\bm{X}} = \inf_{\bm{X}\in (\mathcal{H}_{d,2}^\mathbb{Q})^*}$, $\sup_{\bm{X}}=\sup_{\bm{X}\in (\mathcal{H}_{d,2}^\mathbb{Q})^*}$. Then a simple observation is 
\begin{equation}
\label{3.2}\tag{30}
  \begin{aligned}
   &\sum_{k=1}^n \big(|\bm{\alpha^*}_k \bm{x}|^2 - |\bm{\alpha^*}_k\bm{y}|^2\big)^2 \\&~~~~~~~~~~~\geq \|\bm{xx^*}-\bm{yy^*}\|_F^2\inf_{\bm{X}} \sum_{k=1}^n \big[\Re\big<\bm{X},\bm{\alpha}_k\bm{\alpha^*}_k\big>\big]^2.
  \end{aligned}
\end{equation}
This holds trivially when $\bm{xx^*}=\bm{yy^*}$, and follows from (\ref{3.1}) when $\bm{xx^*}\neq\bm{yy^*}$. We aim to apply small ball method    to  find a positive lower bound for the empirical process $\inf_{\bm{X}}\sum_{k=1}^n  \big[\Re\big<\bm{X},\bm{\alpha}_k\bm{\alpha^*}_k\big>\big]^2$. Firstly,  for any $u>0$ Paley-Zygmund inequality (e.g., Lemma 7.16 in \cite{foucart2013invitation}) yields \begin{eqnarray}\nonumber
Q(u) & := & \inf_{\bm{X}}~ \mathbbm{P}(|\Re\big<\bm{X},\bm{\alpha}_k\bm{\alpha^*}_k\big>|\geq u) \\\nonumber
& =& \inf_{\bm{X}}~ \mathbbm{P}(|\Re\big<\bm{X},\bm{\alpha}_k\bm{\alpha^*}_k\big>|^2\geq u^2) \\\label{equa}
& \geq  & \inf_{\bm{X}}~\frac{(\mathbbm{E} |\Re\big<\bm{X},\bm{\alpha}_k\bm{\alpha^*}_k\big>|^2- u^2)^2}{\mathbbm{E}|\Re\big<\bm{X},\bm{\alpha}_k\bm{\alpha^*}_k\big>|^4}.
\end{eqnarray}
Fix any $\bm{X}\in (\mathcal{H}^\mathbb{Q}_{d,2})^*$ and write it as $\bm{X} = \sigma_1 \bm{ww^*} + \sigma_2\bm{vv^*}$, with $\sigma_1^2 + \sigma_2^2 = 1$, $\sigma_1,\sigma_2\in\mathbb{R}$, $\|\bm{w}\|=\|\bm{v}\|=1$, $\bm{w^*v} = 0$, then 
$$
\Re \big<\bm{X},\bm{\alpha}_k\bm{\alpha^*}_k\big> =  \sigma_1|\bm{w^*\alpha}_k|^2+\sigma_2 |\bm{v^*\alpha}_k|^2. 
$$
By rotational invariance of $\bm{\alpha}_k$ (Lemma \ref{lem3}(a)), $\bm{w^*\alpha_k}$ and $\bm{v^*\alpha}_k$ are independent copies of $\frac{1}{2}\mathcal{N}(0,1)+\frac{1}{2}\mathcal{N}(0,1)\ii + \frac{1}{2}\mathcal{N}(0,1) \jj + \frac{1}{2}\mathcal{N}(0,1) \kk$, hence $4|\bm{w^*\alpha}_k|^2$ follows the $\chi^2$ distribution with 4 degrees of freedom. Thus, \begin{equation}\tag{31}
    \begin{aligned}
    \label{3.4}
       &\mathbbm{E} \big[\Re\big<\bm{X},\bm{\alpha}_k\bm{\alpha^*}_k\big>\big]^2 \\& = (\sigma_1^2 + \sigma_2^2)\mathbbm{E}|\bm{w^*\alpha}_k|^4 + 2\sigma_1\sigma_2 \mathbbm{E}|\bm{w^*\alpha}_k|^2\mathbbm{E}|\bm{v^*\alpha}_k|^2 \\
       &\geq \frac{1}{16} \big\{\mathbbm{E}\big(4|\bm{w^*\alpha}_k|^2\big)^2- \big(\mathbbm{E}~ 4|\bm{w^*\alpha}_k|^2 \big)^2\big\} = \frac{1}{2}.
    \end{aligned}
\end{equation}
We can similarly upper bound the denominator and have
\begin{equation}\tag{32}
    \begin{aligned}
    \label{3.5}
       &\mathbbm{E}|\Re\big<\bm{X},\bm{\alpha}_k\bm{\alpha^*}_k\big>|^4 = \mathbbm{E} \big(\sigma_1|\bm{w^*\alpha}_k|^2+\sigma_2 |\bm{v^*\alpha}_k|^2\big)^4\\&\leq \mathbbm{E} \big( |\bm{w^*\alpha}_k|^2+  |\bm{v^*\alpha}_k|^2\big)^4\\
       & = \frac{1}{4^4}\mathbbm{E}\big(4|\bm{w^*\alpha}_k|^2+4|\bm{v^*\alpha}_k|^2\big)^4  = \frac{2^4\cdot\Gamma(8)}{4^4\cdot\Gamma(4)} = \frac{105}{2}.
    \end{aligned}
\end{equation}
Since (\ref{3.4}), (\ref{3.5}) hold for all $\bm{X}\in (\mathcal{H}^\mathbb{Q}_{d,2})^*$, combining with (\ref{equa}), we have $Q(\frac{1}{2}) \geq \frac{1/16}{105/2} :=c_0.$ 
Secondly, we aim to upper bound $R:= \mathbbm{E}\sup_{\bm{X}}|\frac{1}{n} \sum_{k=1}^n\varepsilon_k \Re\big<\bm{X},\bm{\alpha}_k\bm{\alpha^*}_k\big>|$ with independent Rademacher variable $\varepsilon_k$. Since for any $\bm{A},\bm{B}\in \mathbb{Q}^{d_1\times d_2}$, 
$$
|\big<\bm{A},\bm{B}\big>| \leq \|\bm{A}\|_{nu}\|\bm{B}\| \leq \sqrt{\rank(\bm{A})}\|\bm{A}\|_F\|\bm{B}\|
$$ 
(Lemmas 2.1 and 2.2 in \cite{chen2022color}), we have
$$
R = \frac{1}{n}\mathbbm{E}\sup_{\bm{X}}\left|\Re\Big<\bm{X}, \sum_{k=1}^n\varepsilon_k\bm{\alpha}_k\bm{\alpha^*}_k\Big>\right|\leq \frac{\sqrt{2}}{n} \mathbbm{E}\left\|\sum_{k=1}^n \varepsilon_k\bm{\alpha}_k\bm{\alpha^*}_k\right\|.
$$
We write $ \bm{\alpha}_k  = \bm{\beta}_k + \bm{\gamma}_k\jj $ where $\bm{\beta}_k,\bm{\gamma}_k\in \mathbb{C}^{2d\times 1}$, $2\bm{\beta}$, $2\bm{\gamma}$ have entries i.i.d. drawn from standard complex Gaussian $\mathcal{N}(0,1)+\mathcal{N}(0,1)\ii$. Then some calculations give
\begin{equation}
\nonumber
    \begin{aligned}
    & \mathbbm{E}\left\|\sum_{k=1}^n\varepsilon_k\bm{\alpha}_k\bm{\alpha}^*_k\right\| =  \mathbbm{E}\left\|\sum_{k=1}^n\varepsilon_k(\bm{\beta}_k + \bm{\gamma}_k \jj)(\bm{\beta}_k^* - \jj \bm{\gamma}^*_k)\right\| \\
    & \leq   \mathbbm{E}\Big\|\sum_{k=1}^n \varepsilon_k (\bm{\beta}_k\bm{\beta}^*_k+\bm{\gamma}_k\bm{\gamma}^*_k)\Big\|+   \mathbbm{E}\Big\| \sum_{k=1}^n\varepsilon_k( \bm{\gamma}_k\bm{\beta}_k^\top - \bm{\beta}_k\bm{\gamma}_k^\top)\Big\| \\
    & \leq    \mathbbm{E}\Big\|\sum_{k=1}^n\varepsilon_k\bm{\beta}_k\bm{\beta}^*_k\Big\|+\mathbbm{E}\Big\|\sum_{k=1}^n\varepsilon_k\bm{\gamma}_k\bm{\gamma}^*_k\Big\| +2 \mathbbm{E}\Big\|\sum_{k=1}^n\varepsilon_k\bm{\gamma}_k\bm{\beta}^\top_k\Big\|\\& \leq C_1\sqrt{nd}
    \end{aligned}
\end{equation}
holds for some absolute constant $C_1$. The last inequality is from the estimate   $\mathbbm{E}\|\sum_{k=1}^n \varepsilon_k\bm{\eta}_k\bm{\eta^*}_k\| = O(\sqrt{nd})$ for complex Gaussian $\bm{\eta}_k$ with i.i.d. entries $\mathcal{N}(0,1)+\mathcal{N}(0,1)\ii$ (e.g., Lemma 2.7 in \cite{huang2020performance}), specifically we   let $\bm{\eta}_k ^* = (\bm{\beta}_k^*,\bm{\gamma}_k^\top)$ and then have \begin{equation}
    \nonumber
   \begin{aligned}
   & O(\sqrt{nd}) =  \mathbbm{E}\left\|\sum_{k=1}^n \varepsilon_k\bm{\eta}_k\bm{\eta^*}_k\right\| \\&= \mathbbm{E} \left\|\sum_{k=1}^n \varepsilon_k\begin{bmatrix}\bm{\beta}_k\bm{\beta}_k^* & \bm{\beta}_k\bm{\gamma}_k^\top \\ \overline{\bm{\gamma}_k}\bm{\beta}_k^* & \overline{\bm{\gamma}_k}\bm{\gamma}_k^\top\end{bmatrix}\right\|\geq \mathbbm{E}\Big\|\sum_{k=1}^n\varepsilon_k\bm{\gamma}_k\bm{\beta}^\top_k\Big\|.
   \end{aligned}
\end{equation}

We now apply Mendelson's small ball method, specifically the formulation provided in Proposition \ref{lemma2}. We put all above ingredients together, then with probability at least $1-2\exp(-\frac{1}{2}c_0^2n)$ we have 
$$
\inf_{\bm{X}} \frac{1}{n}\sum_{k=1}^n \big|\Re\big<\bm{X},\bm{\alpha}_k\bm{\alpha^*}_k\big>\big|^2 \geq \frac{1}{16}c_0 -C_1\sqrt{\frac{d}{n}} - \frac{c_0}{32}\geq \frac{1}{64}c_0,
$$
where the last inequality can be guaranteed under the scaling $n\geq Cd$ for some sufficiently large $C$. Plug this into (\ref{3.2}), we obtain $ \sum_{k=1}^n \big(|\bm{\alpha^*}_k \bm{x}|^2 - |\bm{\alpha^*}_k\bm{y}|^2\big)^2 \geq c_2n\|\bm{xx^*}-\bm{yy^*}\|_F^2$. 
Thus, if $|\bm{\alpha^*}_k\bm{x}|^2 =|\bm{\alpha^*}_k\bm{y}|^2  $ for all $k\in [n]$, then $\bm{xx^*}=\bm{yy^*}$. The proof is complete. \hfill $\square$

\setcounter{rem}{3}
\begin{rem}
We stress that extra carefulness is needed to deal with the non-commutativity of quaternion. Specifically,   it is crucial for us to take the real part in   (\ref{3.1}), which enables us to swap the multiplication order by the useful relation $\Re(\Tr(\bm{AB}))=\Re(\Tr(\bm{BA}))$. Note that in general $|\bm{\alpha^*x}|^2   \neq \big<\bm{xx^*},\bm{\alpha\alpha^*}\big>$ when $\bm{\alpha},\bm{x}\in \mathbb{Q}^d$. Indeed, while $\bm{xx^*}$ and $\bm{\alpha\alpha^*}$ are Hermitian, $\big<\bm{xx^*},\bm{\bm{\alpha\alpha^*}}\big>$ may not   even  be real, which is in stark contrast to complex matrices.  For example, we 
can consider $\bm{x} = [1,\ii]^\top$, $\bm{\alpha} = [1,\jj]^\top$ to see
this. 
\end{rem}

\subsection{Proof for Lemma 2}
\setcounter{lem}{1}
\begin{lem}
Under Condition 1, if $\bm{z}_0 \in E_\epsilon(\bm{x})$ and the step size $0<\eta \leq \frac{2}{\beta}$, then the sequence $\{\bm{z}_t\}$ produced by quaternion wirtinger flow satisfies \begin{equation}\tag{33}
    \dist^2(\bm{z}_{t+1},\bm{x}) \leq \Big(1- \frac{2\eta}{\tau}\Big) \dist^2(\bm{z}_t,\bm{x}).
\end{equation}
\end{lem}

\noindent{\it Proof.} We only prove $t = 0$ for (\ref{B.1}), for general $t$ the result follows by noting $\bm{z}_t$ remains in $E_\epsilon(\bm{x})$. We do some estimates as follows \begin{equation}
    \begin{aligned}\nonumber
    &\dist^2(\bm{z}_1,\bm{x}) \leq \|\bm{z}_0 - \eta \nabla f(\bm{z}_0)-\bm{x}\phi(\bm{z}_0)\|^2 \\
    & = \dist^2(\bm{z}_0,\bm{x})+ \eta^2\|\nabla f(\bm{z}_0)\|^2 \\&~~~~~~~~~~~~-2\eta\cdot  \Re\big<\nabla f(\bm{z}_0), \bm{z}_0 - \bm{x}\phi(\bm{z}_0)\big> \\
    & \leq \big(1-\frac{2\eta}{\tau}\big) \dist^2(\bm{z}_0,\bm{x}) +\eta\big(\eta - \frac{2}{\beta}\big) \|\nabla f(\bm{z}_0)\|^2\\& \leq \big(1-\frac{2\eta}{\tau}\big) \dist^2(\bm{z}_0,\bm{x}),
    \end{aligned}
\end{equation}
where we use Condition 1 in the third line. \hfill $\square$ 

\subsection{Proof for Lemma 5}

\setcounter{lem}{4}

\begin{lem}

  Assume $\bm{x}$ is the fixed underlying signal. Given $\delta \in (0,1]$. If $n \geq C_0 \delta^{-2}d\log n$ for sufficiently large hidden constant, then with probability at least $1-C_1n^{-9}-C_2\exp(-C_3d)$, $\bm{z}_0\in E_{2\delta}(\bm{x})$.
  \end{lem}
  \vspace{2mm}

\noindent{\it Proof.} Recall that $\bm{z}_0 = \big(\frac{1}{n}\sum_{k=1}^n |\bm{\alpha^*}_k\bm{x}|^2\big)^{1/2}\cdot \bm{\nu}_{in}$ (where $\|\bm{\nu}_{in}\|=1$ is the eigenvector of $\bm{S}_{in}= \frac{1}{n}\sum_{k=1}^n y_k \bm{\alpha}_k\bm{\alpha^*}_k$   corresponding to the largest standard eigenvalue), and so \begin{equation}
    \begin{aligned}\nonumber
     &\dist(\bm{z}_0,\bm{x}) = \|\bm{z}_0 - \bm{x}\phi(\bm{\nu}_{in}) \|\\&\leq \|\bm{z}_0-\bm{\nu}_{in}\|+ \|\bm{\nu}_{in} - \bm{x}\phi(\bm{\nu}_{in})\|\\
     &\leq {\big|\frac{1}{n}\sum_{k=1}^n|\bm{\alpha^*}_k\bm{x}|^2-1\big|^{\frac{1}{2}}}+\sqrt{2}\sqrt{1-|\bm{\nu}^*_{in}\bm{x}|}: = T_1 + \sqrt{2}T_2.
    \end{aligned}
\end{equation}
 Bernstein's inequality gives $\mathbbm{P}(|T_1|^2 \geq t)
\leq 2\exp(-cn\min\{t,t^2\})$ for any $t>0$. Taking $t = \frac{\delta}{4}$, then with probability at least $1-2\exp(-c\delta^2n/16)$, $T_1\leq \frac{\delta}{2}$.

We then go into $T_2$.   Specifically, we use the real representation of a quaternion matrix and Lemma 8   in the paper. We can show that with probability at least $1-C_1n^{-9}-C_2\exp(-C_3d)$, $\|\bm{S}_{in}-\mathbbm{E}\bm{S}_{in}\|= \|\bm{S}_{in}- \bm{I}_d-\frac{1}{2}\bm{xx^*}\|\leq \frac{\delta}{16}$. Therefore, we have $\frac{\delta}{16}\geq \bm{\nu^*}_{in}\bm{S}_{in}\bm{\nu}_{in} -1-\frac{1}{2}|\bm{\nu^*}_{in}\bm{x}|^2$, and $\frac{\delta}{16}\geq \frac{3}{2}-\bm{x^*S}_{in}\bm{x}$. By construction $\bm{\nu^*}_{in}\bm{S}_{in}\bm{\nu}_{in} \geq \bm{x^*S}_{in}\bm{x}$. Taken collectively, it gives $\frac{\delta}{16}\geq \frac{1}{2}-\frac{\delta}{16} - \frac{1}{2}|\bm{\nu^*}_{in} \bm{x}|^2\geq \frac{1}{2}(1-|\bm{\nu^*}_{in}\bm{x}|)-\frac{\delta}{16}$, hence $T_2\leq \frac{\delta}{2}$. The proof is concluded. \hfill $\square$

\subsection{Proof for Lemma 8}

\setcounter{lem}{7}

\begin{lem}
Assume entries of $\bm{\alpha}\in \mathbb{Q}^d$ are i.i.d. drawn from $\mathcal{N}_\mathbb{Q}$, and $\bm{u},\bm{v}\in\mathbb{Q}^d$ satisfy $\|\bm{u}\| = \|\bm{v}\|=1$. Then we have the following:

\vspace{1mm}

\noindent{\rm	 (a)} {\rm(rotational invariance)} For any unitary matrix $\bm{P}\in \mathbb{Q}^{d\times d}$, $\bm{P\alpha}$ and $\bm{\alpha}$ have the same distribution.

\vspace{1mm}
\noindent{\rm (b)} $\mathbbm{E} |\bm{\alpha^* u}|^{2l} = \frac{(l+1)!}{2^l}$ for any positive integer $l$.

\vspace{1mm}
\noindent{\rm (c)} $\mathbbm{E}\big[\Re(\bm{u^*\alpha\alpha^* v})\big]^2 = \frac{1}{4} + \frac{5}{4}\big[\Re(\bm{u^*v})\big]^2 - \frac{1}{4}\big|\Im(\bm{u^*v})\big|^2$.

\vspace{1mm}
\noindent{\rm (d)} $\mathbbm{E}\big[\bm{\alpha\alpha^*u}|\bm{\alpha^*v}|^2\big] = \bm{u}+\frac{1}{2}\bm{vv^*u}$.
\end{lem}

\noindent{\it Proof.} (a) Recall that $\mathcal{T}_1(\bm{A})$ is defined to be the first column (of blocks) of $\mathcal{T}(\bm{A})$, we only need to show $\mathcal{T}_1(\bm{P\alpha})$ and $\mathcal{T}_1(\bm{\alpha})$ have the same distribution. Noting $\mathcal{T}_1(\bm{P\alpha}) = \mathcal{T}(\bm{P})\mathcal{T}_1(\bm{\alpha})$, and observe that $\mathcal{T}_1(\bm{\alpha})$ has entries i.i.d. drawn from $\frac{1}{2}\mathcal{N}(0,1)$, while $\mathcal{T}(\bm{P})$ is real orthogonal matrix. By rotational invariance of real Gaussian vector, the result is immediate.

\noindent
(b) By rotation invariance, $\bm{\alpha^*u}$ is a realization of $\mathcal{N}_\mathbb{Q}$, hence $4|\bm{\alpha^*u}|^2$ is just the $\chi^2$ distribution with 4 degrees of freedom. Hence, the result follows from the known values of the moments of $\chi^2$ distribution.

\noindent
(c) We find a unitary $\bm{P}$ such that $\bm{Pu}  = \bm{e}_1$, further write $\bm{Pv} = \bm{w} = [\mathtt{w}_i]$, $\bm{P\alpha} = \bm{\gamma} = [\gamma_i]$, where $\bm{\gamma}$ and $\bm{\alpha}$ have the same distribution. Some algebra gives 
\begin{equation}
    \begin{aligned}
    \nonumber
    & \mathbbm{E} \big[\Re(\bm{u^*\alpha\alpha^* v})\big]^2 = \mathbbm{E}\big[\Re(\bm{e}_1^\top\bm{\gamma \gamma^*w})\big]^2 \\&= \mathbbm{E}\big[|\gamma_1|^2\Re(\mathtt{w}_1) + \Re(\gamma_1\sum_{k=2}^d \overline{\gamma}_k \mathtt{w}_k)\big]^2 \\
    & = \mathbbm{E} \big[|\gamma_1|^2 \Re(\mathtt{w}_1)+ \sqrt{1-|\mathtt{w}_1|^2}\cdot\Re(\gamma_1\tilde{\gamma})\big]^2\\&= [\Re(\mathtt{w}_1)]^2\mathbbm{E}|\gamma_1|^4+ (1-|\mathtt{w}_1|^2)\mathbbm{E}[\Re(\gamma_1\tilde{\gamma})]^2\\& = \frac{3}{2}[\Re( \mathtt{w}_1)]^2 + \frac{1}{4}(1-|\mathtt{w}_1|^2) \\& = \frac{1}{4} + \frac{5}{4}\big[\Re(\bm{u^*v})\big]^2 - \frac{1}{4}\big|\Im(\bm{u^*v})\big|^2,
    \end{aligned}
\end{equation}
where we let $\tilde{\gamma}= \sum_{k=2}^d \mathtt{w}_k\overline{\gamma}_k / \sqrt{1-|\mathtt{w}_1|^2}$ in the second line, and invoke the fact that $\tilde{\gamma},\gamma_1$ are independent copies of $\mathcal{N}_\mathbb{Q}$.

\vspace{1mm}

\noindent
(d) We similarly find unitary $\bm{P}$ such that $\bm{Pv}=\bm{e}_1$, and denote $\bm{Pu} = \bm{w}$, $\bm{P\alpha} =\bm{\gamma} = [\gamma_i]$. Note that $\bm{\gamma}$ and $\bm{\alpha}$ have the same distribution. Then, it proceeds that 
\begin{equation}
    \begin{aligned}\nonumber
    & \mathbbm{E}\big[\bm{\alpha\alpha^*u|\bm{\alpha^*v}|^2}\big] = \mathbbm{E}\big[\bm{P^*\gamma \gamma^*w}|\bm{\gamma^*e}_1|^2\big] \\& =\bm{P^*} \big[\mathbbm{E}|\gamma_1|^2\bm{\gamma\gamma^*}\big]\bm{w} =\bm{P^*} \big[\frac{1}{2}\bm{e}_1\bm{e^*}_1+\bm{I}_d\big]\bm{w}\\&=\big(\bm{I}_d+\frac{1}{2}\bm{vv^*}\big)\bm{u},
    \end{aligned}
\end{equation}
the result follows. \hfill $\square$
\vspace{1mm}

\subsection{Proof for Lemma 9}

\setcounter{lem}{8}

\begin{lem}
Assume $\bm{\gamma}_1,\cdots,\bm{\gamma}_n$ are independent random vectors in $\mathbb{R}^d$ that have entries $[\gamma_{kj}]$ i.i.d. drawn from $\mathcal{N}(0,1)$. Fix $i,j\in \{1,2\}$ and any sufficiently small $\delta>0$, when $n\geq C_1 \delta^{-2}d\log n$ for sufficiently large $C_1$, with probability at least $1-2n^{-9}-2\exp(-cd)$, we have \begin{equation}\nonumber
    \Big\|\frac{1}{n}\sum_{k=1}^n \gamma_{ki}\gamma_{kj} \bm{\gamma}_k\bm{\gamma}_k^\top - \mathbbm{E}\big[\gamma_{ki}\gamma_{kj} \bm{\gamma}_k\bm{\gamma}_k^\top\big]\Big\|\leq C\delta.
\end{equation}
\end{lem}

\noindent{\it Proof.} We only deal with $i\neq j$, specifically $i=1,j=2$, since the proof for $i=j$ is parallel and indeed less involved. We define the event $\mathcal{A}= \{\max_{k\in [n]} \max_{i=1,2}|\gamma_{ki}| \leq \zeta:=C_1\sqrt{\log n}\}$. A standard estimate gives that we can pick some $C_1$ such that $\mathbbm{P}(\mathcal{A})\geq \mathbbm{P}(|\gamma_{ki}|\leq \zeta)\geq 1- n^{-9}$ (e.g., Proposition 4, \cite{chen2022color}). Conditional on $\mathcal{A}$, we define the random variables $\tilde{\gamma}_{ki},k\in [n],j\in [2]$ that follow the conditional distribution of $\gamma_{ki}$. Evidently, $\tilde{\gamma}_{ki}$ possesses the p.d.f.   $\frac{\tilde{c}}{\sqrt{2\pi}}\exp(-\frac{x^2}{2})\mathbbm{1}(|x|\leq \zeta)$, where $\tilde{c}\in (1,\frac{1}{1-n^{-9}})$ is the normalization constant. Writing $\bm{\tilde{\gamma}}_k = [\tilde{\gamma}_{k1},\tilde{\gamma}_{k2},\cdots, \tilde{\gamma}_{kd}]^\top$, we proceed by conditioning on the event $\mathcal{A}$, hence we need to show w.h.p $T \leq C\delta$, with $T$ defined and then estimated as\begin{equation}
    \begin{aligned}
    \label{A.7}
     &T:=\Big\|\frac{1}{n}\sum_{k=1}^n\tilde{\gamma}_{k1}\tilde{\gamma}_{k2} \bm{\tilde{\gamma}}_k\bm{\tilde{\gamma}}_k^\top -   \mathbbm{E} \big[\gamma_{k1}\gamma_{k2} \bm{\gamma}_k\bm{\gamma}_k^\top\big]\Big\|\\&\leq  \Big\|\frac{1}{n}\sum_{k=1}^n\tilde{\gamma}_{k1}\tilde{\gamma}_{k2} \bm{\tilde{\gamma}}_k\bm{\tilde{\gamma}}_k^\top -  \mathbbm{E} \big[\tilde{\gamma}_{k1}\tilde{\gamma}_{k2} \bm{\tilde{\gamma}}_k\bm{\tilde{\gamma}}_k^\top\big]\Big\|\\
    &+\Big\|\mathbbm{E}\big[\tilde{\gamma}_{k1}\tilde{\gamma}_{k2} \bm{\tilde{\gamma}}_k\bm{\tilde{\gamma}}_k^\top\big] -  \mathbbm{E} \big[\gamma_{k1}\gamma_{k2} \bm{\gamma}_k\bm{\gamma}_k^\top\big]\Big\| := T_{11}+T_{ 12}.
    \end{aligned}
\end{equation}
We estimate $T_{12}$ first. A simple observation is that all but $(1,2)$-th, $(2,1)$-th entries of $\mathbbm{E}\big[\tilde{\gamma}_{k1}\tilde{\gamma}_{k2} \bm{\tilde{\gamma}}_k\bm{\tilde{\gamma}}_k^\top\big]$, $\mathbbm{E} \big[\gamma_{k1}\gamma_{k2} \bm{\gamma}_k\bm{\gamma}_k^\top\big]$ are zero, and the $(1,2)$-th, $(2,1)$-th entry can be estimated as 
\begin{equation}
    \begin{aligned}
    \nonumber
     &|\mathbbm{E}\tilde{\gamma}_{k1}^2 \mathbbm{E} \tilde{\gamma}_{k2}^2 - \mathbbm{E}{\gamma}_{k1}^2 \mathbbm{E} {\gamma}_{k2}^2|\\&\leq |\mathbbm{E} \tilde{\gamma}_{k2}^2 (\mathbbm{E}\gamma^2_{k1}- \mathbbm{E}\tilde{\gamma}_{k1}^2)|  + |\mathbbm{E} \gamma_{k1}^2 (\mathbbm{E}\tilde{\gamma}^2_{k2} - \mathbbm{E}\gamma^2_{k2})|\\
     &\leq 4|\mathbbm{E}\gamma^2_{k1}- \mathbbm{E}\tilde{\gamma}_{k1}^2| =4\big((\tilde{c}^2-1)\mathbbm{E}\gamma_{k1}^2 \\&~~~~~~~~~~~~~+2 \int_{\zeta}^\infty \frac{x^2}{\sqrt{2\pi}}\exp\big(-\frac{x^2}{2} \big)~\mathrm{d}x\big)\\
     &\leq 4\big(3(\tilde{c}-1) + \int_{\zeta}^\infty \frac{x^3}{\zeta} \exp\big(-\frac{x^2}{2}\big)~\mathrm{d}x  \big) \\&= \frac{12}{n^9-1}+\frac{8}{\zeta}\big(\frac{\zeta^2}{2}+1\big) \exp\big(-\frac{\zeta^2}{2}\big)<\delta.
    \end{aligned}
\end{equation}
In the last inequality we plug in $\zeta = C_1\sqrt{\log n}$ and then can see it  holds under a   $n\geq C_1\delta^{-2}d\log n$ and slightly large $C_1$. Thus, we obtain $T_{12}\leq 2\delta$. We construct $\mathcal{N}_1$ as a $\frac{1}{4}$-net of $\mathbb{S}^{d-1}$ (unit Euclidean sphere of $\mathbb{R}^d$), and we can assume $|\mathcal{N}|\leq 9^d$. A standard covering argument gives (e.g., Chapter 4.1.1, \cite{vershynin2018high})
\begin{equation}
\label{A.9}
   \begin{aligned}
    &T_{11}\leq 2\sup_{\bm{u}\in\mathcal{N}_1}\sup_{\bm{v}\in\mathcal{N}_1} \Big|\frac{1}{n}\sum_{k=1}^n \gamma_{ki}\gamma_{kj} (\bm{u}^\top\bm{\gamma}_k)(\bm{v}^\top\bm{\gamma}_k)\\&~~~~~~~~~~-\mathbbm{E}\big[\gamma_{ki}\gamma_{kj} (\bm{u}^\top\bm{\gamma}_k)(\bm{v}^\top\bm{\gamma}_k)\big]\Big|:=2\widehat{T}_{11}.
   \end{aligned}
\end{equation}
One can also estimate the sub-exponential norm $\|\cdot\|_{\psi_1}$ by using Lemma 2.7.7, \cite{vershynin2018high}, $\|\gamma_{ki}\gamma_{kj} (\bm{u}^\top\bm{\gamma}_k)(\bm{v}^\top\bm{\gamma}_k)\|_{\psi_1} \leq \zeta^2 \| \bm{u}^\top\bm{\gamma}_k\|_{\psi_2}\|\bm{v}^\top\bm{\gamma}_k\|_{\psi_2} \leq C_2 \log n.
  $
So, we can invoke Bernstein's inequality (Theorem 2.8.1, \cite{vershynin2018high}) for fixed $\bm{u},\bm{v}$, then a union bound over $\bm{u},\bm{v}\in \mathcal{N}_1$  yields for some $c$, for all $t>0$, \begin{equation}
\label{A.11}\begin{aligned}
    &\mathbbm{P}(\widehat{T}_{11}\geq t)\\&~~~\leq 2\exp\Big(-c\cdot\big(\frac{n}{\log n}\big) \min\{t,t^2\}+ (4\log 3)d\Big).
    \end{aligned}
\end{equation}
Taking $t = \delta$, it delivers that, as long as $n\geq C_3\delta^{-2} d\log n$, with probability at least $1-2\exp(-\Omega(d))$ we have $\widehat{T}_{11}\leq \delta$, and hence $T_{11}\leq 2\delta$. Putting all pieces together, we conclude that under the sample complexity $n \geq C_3\delta^{-2}d\log n$, with probability at least $(1-\frac{1}{n^9})\cdot (1-2\exp(-\Omega(d)))$, which exceeds $1-2\exp(-\Omega(d)) - 2n^{-9}$, we have $T\leq C\delta$ for some absolute constant $C$. The proof is complete. \hfill $\square$
\begin{rem}
\label{rem5}
Lemma \ref{lem6} plays the same role as Lemma 7.4 in \cite{candes2015phase}, but is more general since it includes the cases of $i\neq j$ (e.g., $(i,j) = (1,2)$). The adopted proof strategy may be of independent technical value. Firstly, We   point out the technical refinement of a unified treatment to $T_{11}$ (\ref{A.9})--(\ref{A.11}). This avoids several estimates from Chebyshev's inequality and produces better probability term. Moreover, our argument is      by conditioning on $\mathcal{A}$ and more rigorous, while \cite{candes2015phase} simply assumed $\mathcal{A}$ without estimating the expectation error term $T_{12}$ in (\ref{A.7}). 
\end{rem}

\subsection{Kodak24 Image Dataset}
\textcolor{black}{
Recall that we compare QPR and two real methods over 24 images from the Kodak24 image dataset. Under the same measurement number, QPR using PQTAF achieves exact recovery in the most images, but does not perform well in the remaining two images. Here, we show these two images in Figure \ref{figfail} to demonstrate the underlying rationale. Specifically, the blocks where quaternion method fails  have (nearly or even exactly) linearly dependent channels and   violate Condition 4 in the paper. In the extreme case where the block has (exactly) linearly dependent channels and hence violates the assumption of Lemma 6 in the paper, there will be identifiability issue and it is impossible to exactly identify the block, so the failure of quaternion method is indeed to be expected.} 

\textcolor{black}{
Furthermore, we compare the three methods   over Kodak24 image dataset again but with the measurement number for each $16\times 16$ block reduced to $6.5\times 256$ and report the results: (1) The performances of the two real methods significantly worsen --- they do not exactly recover any image, and the mean and standard deviation of the 24 PSNR values are respectively $17.49,1.60$ in monochromatic model, and are  $23.12,2.44$ in concatenation model. (2) By contrast, QPR using PQTAF still achieves exact recovery over 22 images, and the PSNR values in the remaining two images (naturally, just the two images in Figure \ref{figfail}) are $23.65$ and $13.82$. More concretely, 8   images reconstructed by the three methods are displayed in Figure \ref{figcopa}.}

\begin{figure}[!t]
    \centering
    \includegraphics[scale = 0.44]{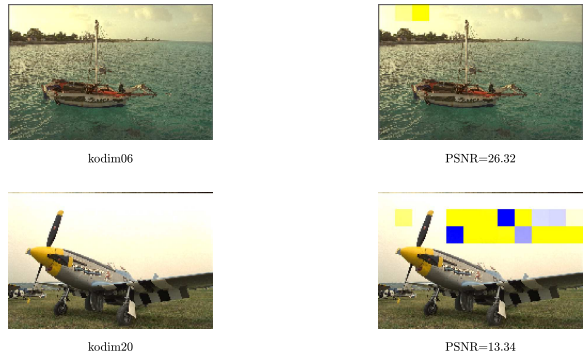}
    \caption{The two images in which QPR fails.}
    \label{figfail}
\end{figure}

\begin{figure}[!t]
    \centering
    \includegraphics[scale = 0.7]{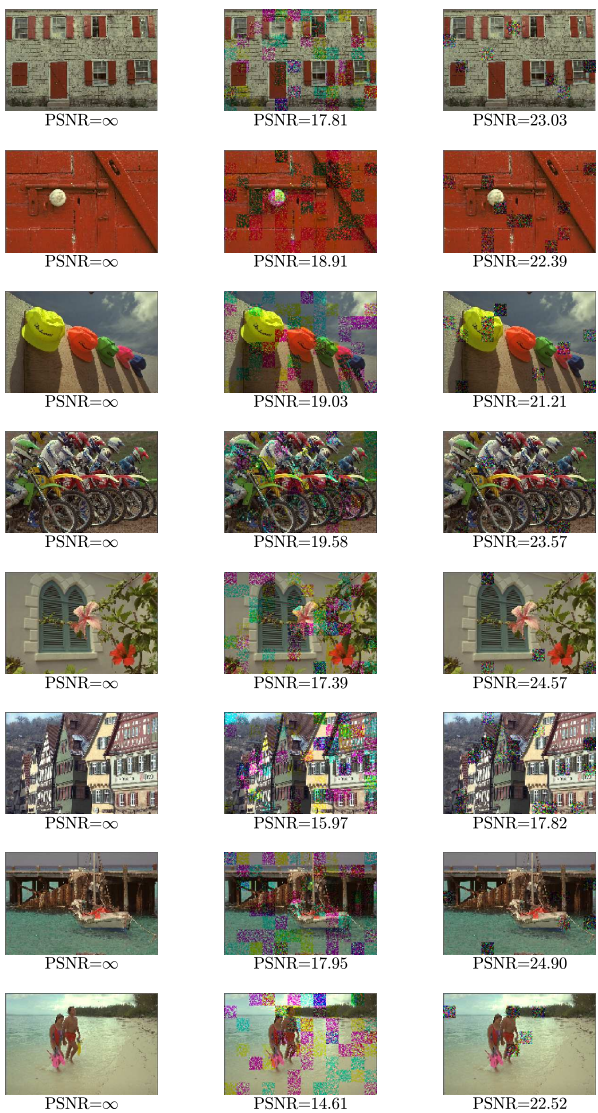}
    \caption{Images reconstructed from "QPR using PQTAF" (col. 1), "monochromatic model using TAF" (col. 2) and "concatenation model using TAF" (col. 3). Images are divided as $16\times 16$ blocks, each of which is independently recovered from $6.5\times 256$ measurements.}
    \label{figcopa}
\end{figure}

 





\end{document}